\newcommand{\version}{October 11, 2010}
\documentclass[11pt,a4paper,english]{article}

\usepackage[bindingoffset=0.3cm,textheight=22.5cm,hdivide={2.7cm,*,2.7cm}, vdivide={*,22cm,*}]{geometry}
\usepackage[dvips,bookmarksnumbered=true,breaklinks=true]{hyperref}

\usepackage{amsmath,amsfonts,amssymb,babel}

\usepackage[numbers,square,comma,sort&compress]{natbib}

\IfFileExists{dsfont.sty}
	{\usepackage{dsfont}
         \newcommand{\id}{\mathds{1}}}
	{\typeout{Package dsfont.sty was not found, using alternative macros.}
         \let\mathds=\mathbb
         \newcommand{\id}{\mbox{1 \kern-.59em {\rm l}}}}
\let\one=\id

\makeatletter

         \@addtoreset{equation}{section}
\makeatother

\let\startappendix=\appendix
\newcommand{\nocontentsline}[3]{}
\newcommand{\tocless}[3]{\bgroup\let\addcontentsline=\nocontentsline#1{#2}#3\egroup}

\newcommand{\Appendix}[1]{
  \refstepcounter{section}
  \section*{Appendix \thesection:\hspace*{1.5ex} #1}
  \addcontentsline{toc}{section}{Appendix \thesection}
}
\newcommand{\SubAppendix}[2]{\tocless\subsection{#1}{\label{#2}}}

\newtheorem{theorem}{Theorem}
\newtheorem{lemma}[theorem]{Lemma}

\newcommand{\qed}{\nobreak \ifvmode \relax \else
      \ifdim\lastskip<1.5em \hskip-\lastskip
      \hskip1.5em plus0em minus0.5em \fi \nobreak
      \vrule height0.75em width0.5em depth0.25em\fi}

\newenvironment{proof}[1][Proof. \hspace*{1ex}]{\begin{trivlist}
\item[\hskip \labelsep {\bfseries #1}]}{\qed\end{trivlist}}

\newcommand{\be}{\begin{equation}}
\newcommand{\ee}{\end{equation}}
\newcommand{\eq}[1]{(\ref{#1})}

\def\nn{\nonumber}
\def\bea{\begin{eqnarray}}
\def\eea{\end{eqnarray}}

\def\beqa{\begin{eqnarray}} 
\def\eeqa{\end{eqnarray}} 
\def\beq{\begin{equation}} 
\def\eeq{\end{equation}}

\def\Tr{{\rm Tr}}

\def\a{\alpha}          
\def\b{\beta}           
\def\d{\delta}    
\def\e{\epsilon}                
        
\def\g{\gamma} 

\def\k{\kappa}
  
\def\m{\mu}     \def\n{\nu}

\def\r{\rho}
\def\s{\sigma}  
\def\t{\tau}
\def\th{\theta}

\def\cJ{{\cal J}}  
\def\cM{{\cal M}} \def\cN{{\cal N}} \def\cO{{\cal O}}
\def\cP{{\cal P}}

\newcommand{\R}{\mathds{R}}

\newcommand{\N}{\mathds{N}}

\def\bit{\begin{itemize}}
\def\eit{\end{itemize}}

\def\({\left(}
\def\){\right)}
\def\diag{\mbox{diag}}

\def\d{\delta}

\def\pa{\partial} \def\del{\partial}

\newcommand{\tr}{\mbox{tr}}

\def\bcomment#1{}

\newcommand{\nc}{non-com\-mu\-ta\-tive}

\newcommand{\eqnref}[1]{Eqn.~(\ref{#1})}		
\newcommand{\secref}[1]{Section~\ref{#1}}		
\newcommand{\appref}[1]{Appendix~\ref{#1}}		
\newcommand{\inv}[1]{\frac{1}{#1}}				

\newcommand{\nabg}{\nabla'}
\newcommand{\nabG}{\nabla}
\newcommand{\lap}{\square_g}
\newcommand{\Lap}{\square_G}
\newcommand{\pb}[2]{\{#1,#2\}}						
\newcommand{\co}[2]{[#1,#2]}						
\newcommand{\aco}[2]{[#1,#2]_+}						
\newcommand{\intg}{\int\!d^4x\sqrt{g}\,}				
\newcommand{\intG}{\int\!d^4x\sqrt{G}\,}				

\newcommand{\intpig}{\int\!\frac{d^4x}{(2\pi)^2}\sqrt{g}\,}				
\newcommand{\intpiG}{\int\!\frac{d^4x}{(2\pi)^2}\sqrt{G}\,}

\renewcommand{\a}{\alpha}
\renewcommand{\b}{\beta}
\renewcommand{\d}{\delta}

\renewcommand{\th}{\theta}

\renewcommand{\r}{\rho}
\renewcommand{\t}{\tau}

\renewcommand{\Xi}{\Xi}

\title{\begin{flushright}
       \small{UWThPh-2010-12}
       \end{flushright}
\vspace{3em}
Curvature and Gravity Actions for Matrix Models II:\\[1ex] the case of general Poisson structures}
\author{Daniel N. Blaschke\footnote{daniel.blaschke@univie.ac.at}~, Harold Steinacker\footnote{harold.steinacker@univie.ac.at}}

\date{\version}
\begin{document}
\maketitle

\begin{center}
\renewcommand{\thefootnote}{\fnsymbol{footnote}}
\textit{Faculty of Physics, University of Vienna\\
Boltzmanngasse 5, A-1090 Vienna (Austria)}
\vspace{0.5cm}
\end{center}%
\begin{abstract}

We study the geometrical meaning of higher-order terms in matrix models of Yang-Mills type
in the semi-classical limit, generalizing recent results \cite{Blaschke:2010rg} 
to the case of 4-dimensional space-time geometries with general Poisson structure. 
Such terms are expected to arise e.g. upon quantization of the IKKT-type  models. 
We identify terms which depend only on the intrinsic geometry and curvature,
including modified versions of the Einstein-Hilbert action, as well as terms
which depend on the extrinsic curvature. 
Furthermore, a mechanism is found which implies that the 
effective metric $G$ on the space-time brane
$\cM \subset \R^D$ ``almost'' coincides with the induced metric $g$. 
Deviations from $G=g$ are suppressed, and characterized by the would-be $U(1)$ gauge field.

\end{abstract}

\newpage
\tableofcontents

\section{Introduction and background}\label{sec:background}

This paper is a continuation of our previous work \cite{Blaschke:2010rg},
where gravitational actions, in particular an analog of the Einstein-Hilbert action, 
were obtained from higher-order terms in matrix models of Yang-Mills type. 

In this framework~\cite{Steinacker:2007dq,Grosse:2008xr,Steinacker:2010rh}, space-time is realized as 
quantized Poisson manifold
$\cM\subset \R^D$ with 
an induced metric $g_{\mu\nu}$ and Poisson tensor $\theta^{\mu\nu}$. These 
structures determine  an effective 
gravitational metric $G^{\mu\nu} = e^{-\s}\theta^{\mu\mu'}\theta^{\nu\nu'} g_{\mu'\nu'}$,
to which matter couples more-or-less 
as in general relativity (GR). 
Since generic 4-dimensional geometries can be 
realized (at least locally) as sub-manifold $\cM \subset \R^{10}$ 
\cite{Friedman:1961}, 
this provides a 
suitable framework for a pre-geometric, ``emergent'' theory of gravity.
As an illustration, a realization of the 
Schwarzschild geometry in this approach is presented in Ref.~\cite{Blaschke:2010ye}.

The dynamics of gravity in this framework and its relation resp. deviation from
general relativity is not yet very well understood. Upon quantization, 
various  higher-order terms are expected to arise in the matrix model, 
or alternatively such terms can be added by hand.
In \cite{Blaschke:2010rg}, we identified a matrix model action 
which in the semi-classical limit reduces to $\intg e^{2\sigma} R[g]$, 
for the most natural case of geometries with $G_{\m\n}=g_{\m\n}$.
However, it turns out that there are several possible matrix actions 
which reduce to the same semi-classical form for $G_{\m\n}=g_{\m\n}$.
Moreover, in order to derive the equations of motion for the geometry,
it is necessary to consider variations which violate this condition. 
In the present paper, we obtain a slightly modified action which for 
coinciding metrics reduces to the Einstein-Hilbert action, and which 
is tensorial (i.e. depends only on the intrinsic geometry of $\cM\subset \R^D$)
for general $G_{\m\n}\neq g_{\m\n}$.
We also identify several other terms which have an intrinsic geometrical meaning.
Some of these terms depend also on the Poisson structure.
There are also ``potential'' terms which may set 
the non-commutativity (NC) scale $e^{-\s}$, as well as
terms which depend on the extrinsic geometry, i.e. the embedding
of $\cM\subset \R^D$.
This should be the beginning of a more systematic study.

An important issue which arises in this context is the role of 
the Poisson or NC structure $\theta^{\mu\nu}$, which in particular determines the difference
$h_{\mu\nu} = G_{\mu\nu} - g_{\mu\nu}$. This Poisson structure can be viewed as
would-be $U(1)$ gauge field, and is governed mainly by the 
``bare'' Yang-Mills term in the matrix model.
We show that this action suppresses 
$h_{\mu\nu}$, and singles out self-dual and anti-selfdual Poisson structures 
with $G_{\mu\nu} = g_{\mu\nu}$
as vacuum solutions. In the case of Minkowski signature, this holds once
a specific complexification of Poisson structures is adopted, which appears 
to be very natural. 
This is important progress in the understanding of emergent gravity in these models, and 
exhibits more clearly the relation with general relativity.

In the present work, we restrict ourselves essentially
to the semi-classical limit of the 
matrix model. Of course, the main appeal for this framework compared with other 
descriptions of gravity is the fact that it goes beyond the 
classical concepts of geometry: Space-time is not put in by hand
but emerges, realized as {\nc} space with an effective geometry,
gauge fields, and matter. Moreover, the IKKT matrix model \cite{Ishibashi:1996xs}
(which is the prime candidate of this class of models with $D=10$)  
can alternatively be viewed as $\cN=4$ supersymmetric 
Yang-Mills gauge theory
on $\R^4_\th$, and hence it is expected to define a good quantum theory.
Therefore these models provide promising candidates for 
a quantum theory of fundamental interactions including gravity.
Moreover, there are several intriguing hints that the role of vacuum energy 
in this framework may be different than in GR.
Nevertheless, much more work remains to be done in order to fully understand
this class of models, and we hope that the current paper provides useful results and tools
for that purpose.

This paper is organized in the following way: We start by reviewing properties 
and important relations of the current framework of matrix models and emergent gravity
in \secref{sec:basic}. This will also fix our notation for the remaining sections. 
We then continue \secref{sec:matrixmodels-intro} by deriving relations for the special case 
of a 4-dimensional embedded manifold $\cM^4\subset\R^D$, and discuss connections and curvature. 
\secref{sec:extensions} will be devoted to higher order extensions to Yang-Mills matrix models 
and their semi-classical limit, whose implications will be discussed in \secref{sec:eom}.

\section{Matrix models and their geometry} 
\label{sec:matrixmodels-intro}

We briefly collect the essential ingredients of the matrix model framework
for emergent gravity, referring 
e.g. to the recent review \cite{Steinacker:2010rh} for more details.

\subsection{Reviewing the basic ingredients}
\label{sec:basic}
The starting point is given by the matrix model of Yang-Mills type, 
\begin{align}
S_{YM}&=-\Tr\co{X^a}{X^b}\co{X^c}{X^d}\eta_{ac}\eta_{bd}\,,
\label{S-YM}
\end{align}
where $\eta_{ac}$ is the (flat) metric of a $D$ dimensional embedding space 
(i.e. $a,b,c,d\in1,\ldots,D$). It can be purely Euclidean, or have one or more time-like directions. 
The ``covariant coordinates'' $X^a$ (cf.~\cite{Madore:2000en}) 
are Hermitian matrices, resp. operators acting on a separable Hilbert space $\mathcal{H}$. 
The commutator of two coordinates will be denoted as
\begin{align}
\co{X^a}{X^b}&=i \th^{ab}\, .
\end{align}
We are interested in configurations which can be interpreted as  $2n$ dimensional
 {\nc} space $\mathcal{M}_\th^{2n}$, in the spirit of non-commutative geometry.
Thus we consider configurations where $2n$ of the matrices (henceforth called $X^\mu$) 
generate a non-commutative algebra interpreted as
{\nc} spaces $\mathcal{M}_\th^{2n}$, and the remaining $D-2n$ matrices are 
(quantized) functions of the $X^\mu$, i.e. functions on $\mathcal{M}_\th^{2n}$.
In other words, 
we split\footnote{More generally, all of the $X^a$
are interpreted as functions on $\mathcal{M}_\th^{2n}$ subject to $D-2n$ relations.
Examples for such NC submanifolds realized by matrix models have been known for a long time,
cf. \cite{Banks:1996nn,Ishibashi:1996xs}.
} the matrices resp. coordinates as
\begin{align}
X^a=\left(X^\m,\phi^i\right)\,,\qquad \m=1,\ldots,2n\,,\qquad i=1,\ldots,D-2n\,,
\end{align}
so that the $\phi^i(X) \sim \phi^i(x)$ in the semi-classical limit define
an embedding of a $2n$ dimensional submanifold
\be
\mathcal{M}^{2n}\hookrightarrow \R^D .
\ee 
Moreover, we can interpret\footnote{In the special case where $\th^{\m\n}$ is constant, this 
leads to {\nc} field theories --- see~\cite{Doplicher:1994tu,Douglas:2001ba} 
for a review of the topic. However, a dynamical commutator seems essential in the context of gravity.}
\begin{align}
\co{X^\m}{X^\n}&\sim i \th^{\m\n}(x)\,
\end{align}
in the semi-classical limit as a Poisson structure on $\mathcal{M}^{2n}$.
Thus we are considering quantized 
Poisson manifolds $(\cM^{2n},\theta^{\mu\nu})$, with quantized
embedding functions $X^a$.
Throughout this paper, $\sim$ denotes the 
semi-classical limit, where commutators are replaced by Poisson brackets.
We will assume that $\th^{\m\n}$ is non-degenerate, so that its 
inverse matrix $\th^{-1}_{\m\n}$ defines a symplectic form 
on $\mathcal{M}^{2n}$.
The sub-manifold $\cM^{2n}\subset\R^D$ 
is equipped with a non-trivial induced metric\footnote{For a related discussion see e.g. \cite{Paston:2007qr}.}
\begin{align}
g_{\m\n}(x)=\pa_\m x^a \pa_\n x^b\eta_{ab}
=\eta_{\m\n}+\pa_\m \phi^i\pa_\n\phi^j\eta_{ij}\,,
\label{eq:def-induced-metric}
\end{align}
via pull-back of $\eta_{ab}$. 
Finally, we define the following quantities \cite{Steinacker:2008ri}:
\begin{align}\label{eq:notation}
G^{\m\n}&=e^{-\s}\th^{\m\r}\th^{\n\s}g_{\r\s}\,, 
      & \eta&=\inv{4}e^\s G^{\m\n}g_{\m\n}\,, \nonumber\\
\r&=\sqrt{\det{\th^{-1}_{\m\n}}}\,, 
       &  e^{-\s}&=\frac{\r}{\sqrt{\det{G_{\m\n}}}}\,.
\end{align}
The last relation gives a unique definition for $e^{-\sigma}$ provided $n>1$, which we assume.
It is easy to see that the kinetic term for scalar fields on  $\cM^{2n}$
is governed by the effective metric $G_{\mu\nu}(x)$, and in fact
the same metric also governs non-Abelian gauge fields 
and fermions  
in the matrix model (up to possible conformal factors), 
so that $G_{\mu\nu}$ {\em must} be interpreted as
gravitational metric. 
Since the embedding $\phi^i$ is dynamical, the model describes a 
theory of gravity realized on dynamically determined submanifolds of $\R^D$.
We also recall that
\be
\Tr\, \phi \sim \int\! \frac{d^{2n}x}{(2\pi)^n}\, \sqrt{G}\, e^{-\sigma} \phi(x)
\ee
in the semi-classical limit, and
note the remarkable identity
\begin{align}
|G_{\mu\nu}(x)| = |g_{\mu\nu}(x)| , \qquad \mbox{2n=4}
\label{G-g-4D}
\end{align}
which holds on 4-dimensional $\cM^4 \subset \R^D$.
It is also useful to define the following tensor 
\be
{\cJ^\mu}_\nu = e^{-\sigma/2} \theta^{\mu\mu'} g_{\mu'\nu} = -e^{\sigma/2} G^{\mu\mu'} \theta^{-1}_{\mu'\nu} 
\label{J-def}
\ee
which satisfies
\bea
{(\cJ^2)^{\mu}}_\rho &=& - G^{\mu \nu} g_{\nu \rho} 
\,, \nn\\
\tr \cJ^2 &=& - 4 e^{-\sigma} \eta \equiv - (gG) 
\,,
\eea
where `$\tr$' denotes the trace over Lorentz indices.

In Ref.~\cite{Blaschke:2010rg}, we focused on the particular case of 4-dimensional geometries with 
\begin{align}
G^{\m\n}&=g^{\m\n} \qquad \rightarrow \qquad \eta=e^\s\, .
\label{eq:Gisg-etaisesigma}
\end{align}
Clearly, this defines an almost-K\"ahler manifold with almost-complex structure $\cJ^2 = -1$. 
For such geometries to be consistent in the case of Minkowski signature, we have to assume that
$\th^{\m\n}$ has imaginary time-like components, which
is natural in view of the correspondence $X^0 \to iT$, as discussed in \cite{Steinacker:2010rh}.
It is not hard to see that this corresponds to $\th^{\m\n}$ being self-dual with respect to the 
metric $g_{\m\n}$ (cf. \secref{sec:4d-relations} and Ref.~\cite{Steinacker:2008ya}). Such $\th^{\m\n}$ indeed exist
for generic geometries\footnote{with suitable technical assumptions, such as 
global hyperbolicity or asymptotic flatness.}.
We then showed that the Einstein-Hilbert action can be obtained by a certain matrix
action \eq{E-H-action-prev}. However,
variations of $\th^{\m\n}$ away
from a self-dual case lead to metric variations 
\be
G_{\mu\nu} = g_{\mu\nu} + h_{\mu\nu} 
\,. \label{G-h-def}
\ee
Therefore, 
in order to derive the equations of motion for both the (embedding) metric
as well as the Poisson structure $\th^{\m\n}$, it is necessary to allow at least small
deviations from $G_{\m\n}=g_{\m\n}$.
We will in fact identify a mechanism in \secref{sec:eom} which generically implies $G \approx g$ 
to a very good approximation, at least for 
geometries with mild curvature. This justifies to consider only 
linearized corrections in $h_{\mu\nu}$, and provides an important step towards 
clarifying the relation with general relativity.

\paragraph{Notation.}
We will adopt the convention that Latin matrix indices are raised and lowered 
with $\eta_{ab}$ throughout this paper (resp. $\d_{ab}$ in the Euclidean case). 
As we consider deviations from the self-dual geometries introduced above, we will inevitably encounter 
two types of covariant derivatives: those with respect to the effective metric $\nabG:=\nabla[G]$, 
and those with respect to the induced metric $\nabg:=\nabla[g]$. We will use this notation throughout the remainder of this paper. 
Furthermore, we will use the abbreviations $(Gg) \equiv G^{\mu\nu}g_{\mu\nu}$ and $(Gg)^\m_\a\equiv G^{\m\r}g_{\r\a}$.

\subsection{Special relations in \texorpdfstring{$2n=4$}{2n=4} dimensions}
\label{sec:4d-relations}

In this section we collect some basic results on the geometry 
of $\cM^4\subset \R^D$ in the presence of the structures defined above. 
We consider the case of general metrics $G_{\mu\nu} \neq g_{\mu\nu}$ 
on  $2n = 4$ dimensional manifolds 
where the tensor ${\cJ^\mu}_\nu$ defined in \eqref{J-def} becomes unimodular, i.e. $\det \cJ = 1$. 
This leads to the existence of a remarkable identity which we will now derive. 
Consider first the Euclidean case.
Since everything is formulated in a tensorial way, 
we can diagonalize the embedding metric at that point
$g_{\mu\nu}|_p = \d_{\mu\nu}$,
and bring the Poisson tensor resp. the symplectic form into canonical form
\be
\omega = \theta^{-1}\,(\a\, dx^0 dx^3 \pm \a^{-1} dx^1 dx^2)
\label{omega-standard-E}
\ee
at $p \in \cM$ using a suitable $SO(4)$ rotation. This leads to
\begin{align}
G^{\mu\nu} 
&= \diag(\a^2,\a^{-2},\a^{-2},\a^2) \qquad\mbox{at}\quad p \in \cM
\,,
\end{align}
and similarly ${\cJ^\mu}_\nu = -\diag(\a^2,\a^{-2},\a^{-2},\a^2)$ at $p \in \cM$.
In particular, it follows that
\be
\frac 14 (Gg) = e^{-\sigma}\eta = \frac 12(\a^2+\a^{-2}) \,\, \geq \,\, 1 \,.
\label{eta-sigma-alpha}
\ee
Furthermore, we obtain the following characteristic equation\footnote{If we would consider 
real $\theta^{\mu\nu}$ in the 
Minkowski case, this relation would be replaced by 
$\cJ^2 + 2 e^{-\sigma} \eta - \cJ^{-2} = 0$.} for $\cJ^2$ 
 \cite{Steinacker:2008ya}:
\bea
{(\cJ^2)^\mu}_\nu + 2 e^{-\sigma} \eta {\d^\mu}_\nu + {(\cJ^{-2})^\mu}_\nu = 0 
\,, \label{4D-id}
\eea
or equivalently
\be
(GgG)^{\mu\nu} = - \left(\cJ^2 G\right)^{\m\n} = 2 e^{-\s} \eta G^{\mu\nu} - g^{\mu\nu}
= \inv{2}(Gg)G^{\mu\nu} - g^{\mu\nu}
\,. \label{4D-id-2}
\ee
Furthermore, observe that $\star (dx^0 dx^3) = dx^1 dx^2$
where $\star$ denotes the Hodge star
defined by $\varepsilon^{\mu\nu\rho\sigma}$ and 
$g_{\mu\nu}$ on $\cM^4$. This means that the corresponding symplectic form
is (anti-) self-dual ((A)SD) if and only if
\be
\star \omega = \pm \omega 
\quad \Leftrightarrow \quad
\a = 1 \,\,\,\mbox{resp.} \,\, e^{-\sigma}\eta = 1 
\quad \Leftrightarrow \quad G_{\mu\nu} = g_{\mu\nu} 
\quad \Leftrightarrow \quad \cJ^2 = -1 
\,,
\label{selfdual-E}
\ee
in which case $\cM^4$ becomes an almost-K\"ahler manifold with almost-complex structure $\cJ$. 
These statements generalize to the case of Minkowski signature, provided
we consider complexified $\theta^{\mu\nu}$ with imaginary 
time-like components $\theta^{0\nu}$, see \cite{Steinacker:2010rh}.

Furthermore, we also note the following useful identity
\be
\del_\a (\rho\theta^{\mu\a}) =0
\ee
which holds in any coordinates, and follows from the Jacobi identity.
On $2n=4$-dimensional branes, it implies
\bea
0 &=& \del_\a (e^{-\sigma} \sqrt{|g|}\theta^{\mu\a} ) = 
\sqrt{|g|}\,\nabg_\a (e^{-\sigma} \theta^{\mu\a}) \nn\\
&=& \del_\a (e^{-\sigma} \sqrt{|G|}\theta^{\mu\a} ) = 
\sqrt{|G|}\,\nabG_\a (e^{-\sigma} \theta^{\mu\a}) 
\label{nabla-theta-id}
\eea
using $|g| = |G|$.
Note furthermore that
\bea
G^{\mu\a}\nabg_\a \theta^{-1}_{\mu\nu} 
&=& \nabg_\a (G^{\mu\a}\theta^{-1}_{\mu\nu}) - \theta^{-1}_{\mu\nu}\nabg_\a G^{\mu\a} \nn\\
&=& -\nabg_\a (e^{-\sigma} \theta^{\mu\a} g_{\mu\nu}) - \theta^{-1}_{\mu\nu}\nabg_\a G^{\mu\a} \nn\\
&=& - \theta^{-1}_{\mu\nu}\nabg_\a G^{\mu\a} 
\label{div-theta-g-id}
\eea
using the basic identity \eq{nabla-theta-id}.

\paragraph{Determinants.}
Consider the scalar function
\bea
\det \cJ = e^{-n\sigma} \det (\theta^{\mu\nu})\det (g_{\mu\nu}) 
\eea
which satisfies $\det \cJ = 1$ in $2n=4$ dimensions. In that case, it follows that
\begin{align}
\del_\a e^{2\sigma} &= \del_\a \det(\theta^{\mu\eta} g_{\eta\nu})
= e^{2\sigma}
g^{\mu\s}\theta^{-1}_{\s\nu} \del_\a( \theta^{\nu\eta}g_{\eta\mu}) \nn\\
&= e^{2\sigma}
\(\theta^{-1}_{\eta\nu} \del_\a \theta^{\nu\eta} + g^{\mu\eta}\del_\a g_{\eta\mu}\)
\,.
\end{align}
We can replace $\del_\a$ with any covariant derivative operator $\nabla_\a$ in this formula.
In particular, for $\nabg=\nabla[g]$ we obtain
\bea
\del_\a e^{2\sigma} 
&=& e^{2\sigma} \theta^{-1}_{\eta\nu} \nabg_\a \theta^{\nu\eta} 
\,.
\eea
Similarly, using ${\cJ^\mu}_\nu = -e^{\sigma/2} G^{\mu\eta}\theta^{-1}_{\eta\nu}$
we get
\bea
\del_\a e^{-2\sigma} 
&=& e^{-2\sigma}  \theta^{\nu\eta} \nabG_\a\theta^{-1}_{\eta\nu}
\,,
\eea
so for $2n=4$ we have
\bea
2\del_\a \sigma &=& \theta^{-1}_{\eta\nu} \nabg_\a \theta^{\nu\eta} 
= \theta^{-1}_{\eta\nu} \nabG_\a \theta^{\nu\eta}\, . 
\eea
Since $\det(G^{\mu\eta} g_{\eta\nu})=1$ in  $2n=4$ dimensions, 
a similar argument yields 
\begin{align}
0 &= \del_\a \det(G^{\mu\eta} g_{\eta\nu}) 
= g^{\mu\s}G_{\s\nu} \del_\a(G^{\nu\eta}g_{\eta\mu}) \nn\\
&= G_{\eta\nu} \del_\a G^{\nu\eta} + g^{\mu\eta}\del_\a g_{\eta\mu}
\,,
\end{align}
and likewise for any covariant derivatives. This implies 
\be
 g^{\mu\eta}\nabG_\a g_{\eta\mu} = 0 = G_{\eta\nu} \nabg_\a G^{\nu\eta}
\,. \label{g-det-nabla-id}
\ee
In the computations of the subsequent sections, we will make use of the important relations \eqref{4D-id-2}, 
\eqref{nabla-theta-id}, \eqref{div-theta-g-id} and \eqref{g-det-nabla-id} in many places.

\subsection{Intrinsic curvature.}

Since we consider general geometries $G_{\m\n}\neq g_{\m\n}$ in this paper, 
we will inevitably encounter the tensor 
\bea
C_{\a;\mu\nu} := \del_\a x^a \nabG_\mu\del_\nu x_a = \frac 12 
\(\nabG_\mu g_{\nu\a} + \nabG_\nu g_{\mu\a} - \nabG_\a g_{\mu\nu}\)
\,,
\label{C-tensor}
\eea
in subsequent computations. 
Contracting this tensor with $G^{\m\n}$, one derives 
\begin{subequations}\label{del-g-id-both}
\begin{align}
\del_\a x^a \Lap x_a 
&\stackrel{\phantom{2n=4}}{=} \nabG_\mu(G^{\mu\nu}g_{\nu\a}) - 2 \del_\a (e^{-\sigma}\eta)
 = \nabG^\nu g_{\nu\a} - \frac 12 \del_\a (gG) 
\,, \label{del-g-id-1}\\
&\stackrel{2n=4}{=} -G_{\a\nu}\nabG_\mu g^{\mu\nu} \,.
\label{dXboxX-4D} \\
\del_\a x^a \nabG_\mu\del^\a x_a 
&\stackrel{\phantom{2n=4}}{=} \inv{2} \pa_\m(Gg) 
\,, \label{del-g-id-2}
\end{align}
\end{subequations}
where the 4D identity \eq{4D-id-2} is used in \eq{dXboxX-4D} and ``l.h.s.$\stackrel{2n=4}{=}$r.h.s.'' denotes equality iff $2n=4$.  

Keeping these relations in mind, we now derive the curvature tensor 
with respect to the metrics $G_{\m\n}$ and $g_{\m\n}$: 
For a general embedding $\cM \subset \R^D$ with Cartesian embedding functions $x^a: \cM\hookrightarrow \R^D$, 
consider the expression
\bea
&& 
\nabla_\sigma\nabla_\mu x^a\nabla_\rho\nabla_\nu x_a - \nabla_\sigma\nabla_\nu x^a\nabla_\mu\nabla_\rho x_a \nn\\
&=& \nabla_\sigma(\nabla_\mu x^a\nabla_\rho\nabla_\nu x_a) -\nabla_\mu x^a\nabla_\sigma\nabla_\rho\nabla_\nu x_a 
 - \nabla_\rho(\nabla_\sigma\nabla_\nu x^a\nabla_\mu x_a)
+\nabla_\rho\nabla_\sigma\nabla_\nu x^a\nabla_\mu x_a \nn\\
&=& \nabla_\sigma C_{\mu;\rho\nu} - \nabla_\rho C_{\mu;\s\nu}
+[\nabla_\rho,\nabla_\sigma]\nabla_\nu x^a\nabla_\mu x_a \nn\\
&=&  \nabla_\sigma C_{\mu;\rho\nu} - \nabla_\rho C_{\mu;\s\nu}
+{(Gg)^{\eta}}_{\mu}  R_{\rho\sigma\nu\eta}[G] 
\,. \label{R-dX-formula}
\eea
Unless stated otherwise, we will always understand $R_{\rho\sigma\nu\eta} \equiv R_{\rho\sigma\nu\eta}[G]$ throughout this paper.
All the terms in \eqref{R-dX-formula} are tensorial, and we obtain
\begin{align}\label{eq:def-RG}
(Gg)^{\eta}_{\mu} R_{\rho\sigma\nu\eta}[G]
&= \nabla_\sigma\nabla_\mu x^a\nabla_\rho\nabla_\nu x_a 
- \nabla_\sigma\nabla_\nu x^a\nabla_\mu\nabla_\rho x_a 
 -\nabla_\sigma C_{\mu;\rho\nu} + \nabla_\rho C_{\mu;\s\nu}
\,. 
\end{align}
Repeating this calculation with $\nabG$ replaced by the covariant derivative with respect to the induced metric 
$\nabla[g]=\nabg$, we recover the Gauss-Codazzi theorem due to $\nabg_\m x^a\nabg_\r\nabg_\n x_a=0$:
\begin{align}
R_{\r\s\n\m}[g] &=g_{\m\t}{R[g]_{\r\s\n}}^\t=\nabg_\s\nabg_\m x^a\nabg_\r\nabg_\n x_a - \nabg_\s\nabg_\n x^a\nabg_\m\nabg_\r x_a 
\,. \label{eq:def-Rg}
\end{align}
For the self-dual case $C_{\mu;\rho\nu} = \nabla_\mu x^a\nabla_\rho\nabla_\nu x_a=0$, and 
both curvature tensors \eqref{eq:def-RG} and \eqref{eq:def-Rg} coincide.

\paragraph{Relating $R[g]$ and $R[G]$.}
The covariant derivatives $\nabG_\mu$ and $\nabg_\mu$ are related 
via the tensors $C_{\a;\mu\nu}$ as follows:
\bea
\nabg_\mu V_\nu &= & \nabG_\mu V_\nu -  C_{\a;\mu\nu} g^{\a\b} V_\b  
 =  \nabG_\mu V_\nu + \tilde C_{\a;\mu\nu} G^{\a\b} V_\b  
\,,
\label{rel-cov-der}
\eea
for some vector $V_\n$, and 
where $\tilde C_{\a;\mu\nu}$ is defined by 
replacing $g$ with $G$ (and hence $\nabG$ with $\nabg$) in \eq{C-tensor}. This implies
\bea
&& g^{\a\b} C_{\a;\mu\nu} = \frac 12 g^{\a\b}
\(\nabG_\mu g_{\nu\a} + \nabG_\nu g_{\mu\a} - \nabG_\a g_{\mu\nu}\) \nn\\
&=& -G^{\a\b}\tilde C_{\a;\mu\nu} = - \frac 12 G^{\a\b}
\(\nabg_\mu G_{\nu\a} + \nabg_\nu G_{\mu\a} - \nabg_\a G_{\mu\nu}\) 
\,,
\label{C-tensor-G}
\eea
which has a number of useful consequences:
\bea\label{cons-g-1}
g^{\a\mu} C_{\a;\mu\nu} &=& \frac 12 g^{\a\mu}\nabG_\nu g_{\mu\a} 
= -G^{\a\mu}\tilde C_{\a;\mu\nu} = - \frac 12 G^{\a\mu} \nabg_\nu G_{\mu\a}  
 =0 \,, \nn\\
 g^{\a\b} g^{\mu\nu} C_{\a;\mu\nu} &\stackrel{2n=4}{=}&  g^{\a\b}g^{\mu\nu}\nabla_\mu g_{\nu\a}  
 = - \nabG_\mu g^{\mu\b}   \nn\\
&=& -  G^{\a\b}g^{\mu\nu} \nabg_\mu G_{\nu\a} 
 + \inv2 G^{\a\b} \del_\a (g^{\mu\nu} G_{\mu\nu}) \,,\nn\\
 g^{\a\b} G^{\mu\nu} C_{\a;\mu\nu} &=&  g^{\a\b}G^{\mu\nu}\nabG_\mu g_{\nu\a} 
 -\inv2 g^{\a\b}\pa_\a(Gg) \nn\\
&\stackrel{2n=4}{=}& -  G^{\a\b}G^{\mu\nu} \nabg_\mu G_{\nu\a} 
 = \nabg_\mu G^{\mu\b} 
\,,
\eea
where we have used \eqref{g-det-nabla-id}. 
Furthermore, we may define projectors on the tangential resp. normal bundle of $\cM \subset \R^D$ as
\begin{align}
\cP_T^{ab}&=g^{\m\n}\pa_\m x^a\pa_\n x^b\,, & \cP_N^{ab}&=\eta^{ab}-\cP_T^{ab}\,.
\label{eq:projectors}
\end{align}
Hence, by the very definition of the covariant derivative associated to $g_{\mu\nu}$, 
we have 
\begin{align}
\nabg_\sigma\nabg_\nu x^a &=\nabla_\sigma\nabla_\nu x^a - g^{\a\b}  C_{\b;\sigma\nu}\del_\a x^a \nn\\
&= \nabla_\sigma\nabla_\nu x^a - g^{\a\b}\del_\a x^a \del_\b x^b \nabla_\sigma\nabla_\nu x_b \nn\\
&= \cP_N^{ab} \nabla_\sigma\nabla_\nu x_b 
\,. \label{eq:app-rel-lap-Lap-x}
\end{align}
This allows to relate the curvature tensors\footnote{cp. also~\cite{Wald:1984}.} 
associated to $G_{\mu\nu}$ resp. $g_{\mu\nu}$:
\begin{align}
R_{\rho\sigma\nu\mu}[g]
&= \nabg_\sigma\nabg_\mu x^a\nabg_\rho\nabg_\nu x_a 
- \nabg_\sigma\nabg_\nu x^a\nabg_\mu\nabg_\rho x_a \nn\\
&= \cP_N^{ab}\nabla_\sigma\nabla_\mu x_a\nabla_\rho\nabla_\nu x_b
- \cP_N^{ab}\nabla_\sigma\nabla_\nu x_a\nabla_\mu\nabla_\rho x_b \nn\\
&= (Gg)^{\eta}_{\mu} R_{\rho\sigma\nu\eta}[G] 
+\nabla_\sigma C_{\mu;\rho\nu} - \nabla_\rho C_{\mu;\s\nu}
- C_{\a;\sigma\mu} C_{\b;\rho\nu} g^{\a\b}
+ C_{\a;\sigma\nu} C_{\b;\mu\rho} g^{\a\b} 
\,, \nn\\
R_{\rho\nu}[g] 
&= R_{\rho\nu}[G]
+ g^{\s\mu}\nabla_\sigma C_{\mu;\rho\nu} - g^{\s\mu} \nabla_\rho C_{\mu;\s\nu}
-  g^{\s\mu}C_{\a;\sigma\mu} C_{\b;\rho\nu} g^{\a\b}
+  g^{\s\mu}C_{\a;\sigma\nu} C_{\b;\mu\rho} g^{\a\b} 
\,, \label{RG-Rg-relation}
\end{align}
using \eq{eq:projectors} and \eq{R-dX-formula}.
The last terms can be evaluated using 
{\allowdisplaybreaks
\begin{subequations}\label{eq:C-relations}
\bea
g^{\a\b} C_{\a;\sigma\nu} C_{\b;\mu\rho}g^{\rho\nu} g^{\sigma\mu} 
&=& - \frac 34 g^{\rho\nu}\nabla_\nu g^{\b\mu} \nabla_\rho g_{\mu\b}
  - \frac 12 g_{\rho\mu}\nabla_\b g^{\rho\nu}\nabla_\nu g^{\b\mu} 
\,, \label{gCCgg}\\
g^{\a\b} C_{\a;\sigma\mu} g^{\sigma\mu} C_{\b;\rho\nu}g^{\rho\nu}
&\stackrel{2n=4}{=}&  g_{\b\nu} \nabla_\a  g^{\a\b} \nabla_\rho g^{\rho\nu} 
\,,\label{gCgCg}\\
 G^{\a\b} C_{\a;\sigma\nu} C_{\b;\mu\rho}G^{\rho\nu} G^{\sigma\mu} 
&\stackrel{2n=4}{=}&  4 \del_\nu(e^{-\s}\eta)\del^\nu (e^{-\s}\eta) 
 + 2\del_\a (e^{-\s}\eta)  \nabla_{\mu}g^{\mu\a} \nn\\*
 &&  - \frac 34\nabla_\nu g^{\b\mu}\nabla^\nu g_{\mu\b} 
- \frac 12G_{\mu\b}\nabla_\a g^{\mu\rho}\nabla_\rho g^{\a\b} 
\,, \label{CC-GGG}\\
 g^{\s\mu}\nabla_\sigma C_{\mu;\rho\nu} - g^{\s\mu} \nabla_\rho C_{\mu;\s\nu} 
&=& \frac 12 g^{\s\mu}\nabla_\sigma (\nabla_\rho g_{\mu\nu}+\nabla_\nu g_{\rho\mu}-\nabla_\mu g_{\rho\nu})
- \frac 12 g^{\s\mu} \nabla_\rho \nabla_\nu g_{\s\mu}\nn\\*
&=&  \inv2 \Big(-\nabla_\rho\nabla^\mu h_{\mu\nu} -R_{\r\b}[g] h^{\b\a}g_{\a\nu} 
 + (\rho\leftrightarrow\nu) \Big) \nn\\*
&& + \frac 12 \Box_g h_{\rho\nu}  + R_{\a\rho\b\nu}[g] h^{\a\b} + \cO(h^2) 
\,, \label{nablaC}
\eea
\end{subequations}
as derived in \appref{app:CC-GGG}. 
Hence to leading order in $h_{\m\n} = G_{\mu\nu} - g_{\mu\nu}$, we have 
}
\begin{align}
R_{\rho\nu}[g] 
&=  R_{\rho\nu}[G]
-\frac 12 \Big(\nabla_\rho\nabla^\mu h_{\mu\nu}  +R_{\rho\b}[g] h^{\b\a}g_{\a\nu} 
 + (\rho \leftrightarrow \nu)\Big) 
 +\frac 12 \Box_g h_{\rho\nu} 
+ R_{\a\rho\b\nu}[g] h^{\a\b} \nn\\
&\quad + \cO(h^2),\nn\\
R[g] &= R_{\rho\nu}[G] g^{\rho\nu}
- \nabla^\nu\nabla^\mu h_{\mu\nu} \,\, + \cO(h^2)\,, \nn\\
R[G] &= R_{\rho\nu}[g] G^{\rho\nu}
+ \nabla^\nu\nabla^\mu h_{\mu\nu} \,\, + \cO(h^2)\,.
\label{RG-Rg-relation-h}
\end{align}

\subsection{Cartesian tensors}

Now consider the following expressions, which play an important role in the following:
\bea
H^{ab} &=& \inv{2}\aco{\co{X^a}{X^c}}{\co{X^b}{X_c}}  \quad 
\sim\,\, -e^\sigma G^{\mu\nu}\del_\mu x^a \del_\nu x^b 
\,, \nn\\
H &=&  H^{ab} \eta_{ab} =  \co{X^c}{X^d}\co{X_{c}}{X_{d}}\, 
\,\,\sim \,\, -e^\sigma G^{\mu\nu} g_{\mu\nu} = -4 \eta(x) 
\,. \label{eq:def-H}
\eea
The matrix ``energy-momentum tensor'' is then defined by \cite{Steinacker:2008ri}
\begin{align}
T^{ab}&= H^{ab} - \frac 14\eta^{ab} H\, \quad
\sim\,\, \eta\eta^{ab}  - e^\sigma G^{\mu\nu}\del_\mu x^a \del_\nu x^b 
\,.\label{eq:def-T}
\end{align}
It is instructive to consider the projectors defined in \eqnref{eq:projectors} acting on these expressions in the semi-classical limit, i.e. $(\cP_T H)^{ab}\sim H^{ab}$ and $(\cP_N T)^{ab}\sim \eta \cP_N^{ab}$. 
In the special case of $g_{\mu\nu}=G_{\mu\nu}$, the semi-classical limit of the energy-momentum tensor becomes truly related to the projectors: 
\begin{align}
T^{ab} \sim e^\sigma \cP_N^{ab}\,, \qquad \textrm{and }\quad H^{ab} \sim -e^\sigma \cP_T^{ab}
\,. \label{eq:relation-TH-PNT}
\end{align}
Moreover, then
\begin{align}
T^{ab}\Box X_a \Box X_b - \frac 12 T^{ab}\Box H_{ab}\,\,
&\sim\,\,  e^{3\s}R 
\,,\label{E-H-action-prev}
\end{align}
as shown\footnote{The derivation given in \cite{Blaschke:2010rg} for $\intg  e^{2\s}R$ 
also applies without the integral resp. trace.} in \cite{Blaschke:2010rg}. 
However, there are several similar matrix actions which for 
$g_{\m\n}=G_{\m\n}$ reduce to the same semi-classical form. 
It turns out that for general $g_{\mu\nu} \neq G_{\mu\nu}$, 
which we study in the present paper, the left-hand side of 
\eq{E-H-action-prev} is no longer intrinsic, 
i.e. it depends also on the embedding $\cM\subset \R^D$. This makes the derivation of the equations 
of motion more difficult.
However, we will identify a slightly modified matrix action which is 
intrinsic for general geometries in the semi-classical limit. 

Before we continue, let us add a brief remark concerning $H^{ab}$ in $2n=4$ dimensions: 
The 4D identity \eq{4D-id-2} implies
\begin{align}
(H^3)^{ad} -\frac 12 H (H^2)^{ad} + e^{2\sigma}H^{ad} 
&\stackrel{\phantom{2n=4}}{\sim} -e^{2\s}(Gg)^\m_\r\left(e^\s (GgG)^{\r\n} 
- 2\eta  G^{\r\n} + e^{\s}g^{\r\n}\right)\pa_\m x^a \pa_\n x^d \nn\\
&\stackrel{2n=4}{\sim} 0
\,.
\end{align}
This means that $e^{-\sigma}H^{ab}$ has 3 eigenvalues 
$\{0,\a^2,\a^{-2}\}$ with $e^{-\sigma}\eta = \frac 12(\a^2+\a^{-2})$ 
and $H\sim -4\eta$ 
(cf. \secref{sec:4d-relations} and Ref.~\cite{Steinacker:2008ya}). Hence the last relation essentially
characterizes the 4-dimensional nature of $\cM^4$, and
it also encodes the reality structure of $\theta^{\mu\nu}$ at the 
matrix level because it is non-linear.

\paragraph{Semi-classical limit of the tangential conservation law.}

The following useful results for various Poisson brackets are essentially obtained in \cite{Steinacker:2008ya}:
Since $H^{ab}$ is a scalar field on $\cM\subset \R^D$, 
we have\footnote{Notice, that we use the same symbols $H^{ab}$ and $T^{ab}$ for their respective semi-classical 
limits whenever it is clear from context what is meant.} 
\bea
\pb{x_a}{H^{ab}} 
 &=& -\theta^{\mu\nu}\del_\mu x_a \nabla_\nu(e^\sigma G^{\a\b}\del_\a x^a \del_\b x^b) \nn\\
&=& - e^\sigma G^{\a\b}\(\del_\nu\sigma g_{\mu\a}\theta^{\mu\nu} \del_\b x^b
 + \theta^{\mu\nu}\nabla_\nu g_{\mu\a} \del_\b x^b
+ g_{\mu\a}\theta^{\mu\nu} \nabla_\nu \del_\b x^b  \) \nn\\
&=& - G^{\a\b}\theta^{\mu\nu}\nabla_\nu (e^\sigma g_{\mu\a}) \del_\b x^b 
\,.
\eea
This is again tensorial, and can be written in 
a number of different ways:
\bea
\{x_a,H^{ab}\} &=&  -e^\sigma G^{\a\b}\nabla_\nu (\theta^{\mu\nu} g_{\mu\a} ) \del_\b x^b \nn\\
&=& -e^\sigma G^{\a\b}\nabla^\mu (e^\sigma \theta^{-1}_{\mu\a}) \del_\b x^b \nn\\
&=&  \(\partial_\a\eta - e^\sigma\nabla^\rho g_{\rho\a}-2\eta\partial_\a \sigma\) \theta^{\b\a}\del_\b x^b \nn\\
&=&  \(e^\sigma\Box_G x^a \del_\a x^a + \del_\a \eta\) \theta^{\a\b}\del_\b x^b 
\label{X-H-semiclass}
\eea
using the identity \eqref{nabla-theta-id} and
\bea
\theta^{\nu\mu}\partial_\mu\eta 
&=&  e^{\sigma}\nabla_\mu (G^{\mu\mu'}g_{\mu'\nu'} \theta^{\nu'\nu})
+ e^\sigma\theta^{\nu\a} \nabla^\rho g_{\rho\a}+2\eta\theta^{\nu\mu} \partial_\mu \sigma
\eea
which follows from the Jacobi identity \cite{Steinacker:2008ya}.
Together with \eq{del-g-id-1}, we obtain
\be
\{x_a,T^{ab}\} =  e^\s\(\Lap x_a \del_\a x^a\) \th^{\a\b}\del_\b x^b 
\label{T-cons-poisson}
\ee
which also follows directly from the matrix identity \eq{Ward-cons}.
For Yang-Mills matrix models, the tangential conservation law $[X_a,T^{ab}] =0$ 
holds in fact at the matrix level \cite{Steinacker:2008ri} as a consequence of the symmetry $X^a \to X^a + c^a \one$. 
However, higher order terms in the matrix model as considered below
may modify this relation. 
Note also that for 4-dimensional branes, \eq{dXboxX-4D} implies
\be
\pb{x_a}{T^{ab}} \stackrel{2n=4}{=} - e^\s \nabG_\mu g^{\mu\nu} \,G_{\nu\a}\th^{\a\b}\del_\b x^b 
\,, 
\ee
so that the tangential conservation law is equivalent to $\nabla_\mu g^{\mu\nu} =0$.

\paragraph{Exact matrix identities.}

The above semi-classical conservation law \eq{T-cons-poisson} 
can also be obtained from the following matrix identities:
\bea
\,[X_a,H^{ab}] &=& \inv{2} \(\aco{\Box X_c}{ [X^b,X^c]}  
+ \frac 12 [X^b,H]\)  
\,,\nn\\
\,[X_a,T^{ab}] &=& \inv{2}\aco{\Box X_c}{ [X^b,X^c]} \,.
\label{Ward-cons}
\eea

\section{Extensions of the matrix model action}\label{sec:extensions}

We now want to consider more general terms in the matrix model,
which in general have the form
\be
S_P[X] = \Tr (X^{a_1} \ldots X^{a_l}) P_{a_1 \ldots a_l} \,,
\ee
where $P_{a_1 \ldots a_l}$ is an invariant tensor of $SO(D)$ (resp. $SO(1,D-1)$ etc. in the case of Minkowski 
signature). 
Imposing also translational invariance $X^a \to X^a + c^a \one$, 
only terms built out of commutators are admissible.
We will organize such polynomial terms in the matrix model according to 
the power $\ell$ of matrices $X^a$, as well as the number $d$ of commutators.
It is clear that translational invariance implies $d \geq \ell/2$, and that $k = d - \ell/2$
corresponds to the number of derivatives of geometrical tensors such as $\theta^{\mu\nu}$
in the semi-classical limit.
It is thus natural to consider an expansion in $k$ as well as $\ell$.

\subsection{Matrix operators}

Before diving into the possible extensions to the matrix model action, we collect some basic ``building blocks'' 
for which we derive the following semi-classical results:
\begin{lemma}
For any matrices $\Phi \sim \phi(x), \,\,\Psi \sim \psi(x)$, we have
{\allowdisplaybreaks
\begin{subequations}
\bea
\eta^{ab}[X_a,\Phi][X_b,\Psi] &\sim&- e^{\sigma}G^{\mu\nu}\del_\mu\phi\,\del_\nu\psi  
\,, \label{metric-G}\\
\Box \Phi \equiv [X^a,[X_a,\Phi]]  &\sim& -\pb{x^b}{\pb{x^{c}}{\phi}} \eta_{bc}
= -e^\sigma \Box_G \phi
\,, \label{laplace-operator}\\[0.5ex]
H^{ab}[X_a,\Phi][X_b,\Psi] &\sim& e^{2\sigma}(GgG)^{\mu\nu}\del_\mu\phi\,\del_\nu\psi  
\,, \label{HXpXp}  \\
H^{ab} [X_a,[X_b,\Phi]] 
&\sim& e^{2\sigma} (G g G)^{\b\eta}\nabla_\b\del_\eta\phi
 + e^\sigma  \del_\b e^{\sigma} (G g G)^{\eta\b}  \del_\eta\phi  \nn \\* && 
 + \inv4 e^{2\s}  (\del_\rho (Gg) - (Gg) \del_\rho\sigma)G^{\eta\rho} \del_\eta\phi   \nn\\*
&\stackrel{g=G}{\sim}&
e^{2\sigma} \Lap\phi 
\,. \label{HXX-operator} 
\eea
In particular, for $2n=4$-dimensional branes, we have
\bea
(H^{ab} - \frac 12 H\eta^{ab}) [X_a,\Phi][X_b,\Psi]  
&\sim& - e^{2\sigma} g^{\mu\nu}\del_\mu\phi\,\del_\nu\psi  
\,, \label{HXpXp-4d}  \\
\,[X_a,\big(H^{ab} - \frac 12 H \eta^{ab}\big)  [X_b,\Phi]]
&\sim& - e^{2\sigma} (\Box_g \phi+ g^{\mu\nu}\del_\mu \sigma\del_\nu \phi )
\,. \label{HXX-operator-4d} 
\eea
\end{subequations}
}
\end{lemma}

\begin{proof}
Relations \eq{metric-G} and \eq{laplace-operator} are by now well-known \cite{Steinacker:2008ya}, 
and \eq{HXpXp} can be computed straightforwardly as 
\bea
H^{ab}[X_a,\Phi][X_b,\Psi] 
&\sim& e^\sigma G^{\mu\nu} \del_\mu x^a \del_\nu x^b 
\theta^{\a\b}\del_\a x_a \del_\b \phi \theta^{\a'\b'}\del_{\a'} x_b \del_{\b'} \psi \nn\\ 
&=& e^{2\sigma}(GgG)^{\mu\nu}\del_\mu\phi\,\del_\nu\psi 
\,. 
\eea
Now \eq{HXX-operator} can be shown either by a direct computation which is given in \appref{app:HXX-op}, or 
more elegantly by considering the following bilinear form
\begin{align}
\Tr \(\Phi_1 H^{ab}[X_a,[X_b,\Phi_2]]\) &=
\Tr \(- [X_a,H^{ab}][X_b,\Phi_2]\Phi_1 - H^{ab}[X_b,\Phi_2] [X_a,\Phi_1]\) 
\label{HXX-bili}
\end{align}
for any matrices $\Phi_i\sim\phi_i(x)$. 
The first term vanishes for self-dual $\theta$
(up to $\cO(h^2)$, resp. is easy to evaluate), and reads
\begin{align}
& \Tr \big([X_a,H^{ab}][X_b,\Phi_2]\Phi_1 \big)
\sim -\intpiG e^{-\sigma}\phi_1(e^{\sigma}\Lap x^c \del_\a x_c + \del_\a\eta) 
\theta^{\a\b}\del_\b x^b \theta^{\mu\nu}\del_\mu x_b\del_\nu\phi_2 \nn\\
&= \intpiG e^\s \phi_1 \(\nabla^\b g_{\b\a} - \frac 14 \del_\a(gG)+\frac 14\del_\a\sigma (gG)\) G^{\a\nu}\del_\nu\phi_2 
\,,
\end{align}
using \eq{X-H-semiclass} and \eq{del-g-id-1}. 
The second term of \eqref{HXX-bili} can be computed using \eq{HXpXp} yielding
\begin{align}
&\Tr \big( H^{ab}[X_b,\Phi_2] [X_a,\Phi_1] \big)
\sim \intpiG e^\s (GgG)^{\mu\nu} \del_\mu \phi_2 \del_\nu \phi_1 \nn\\
&= -\intpiG \phi_1\(e^{\s}\nabla_\nu\sigma (GgG)^{\mu\nu}\del_\mu \phi_2
+ e^{\sigma} \nabla^\rho g_{\rho\eta}G^{\mu\eta}\del_\mu \phi_2
 + e^{\sigma} (GgG)^{\mu\nu}\nabla_\nu \del_\mu \phi_2\) 
\,. 
\end{align}
Hence 
\bea
\Tr \big(\Phi_1 H^{ab}[X_a,[X_b,\Phi_2]]\big) &=&
\Tr \big(- [X_a,H^{ab}][X_b,\Phi_2]\Phi_1 - H^{ab}[X_b,\Phi_2] [X_a,\Phi_1]\big) \nn\\
&\sim& \intpiG e^\sigma \phi_1  
\Big(\inv{4}\left(\del_\a(gG) -\del_\a\sigma (gG)\right)G^{\a\nu}\del_\nu\phi_2 
\nn\\
&& \quad
+\nabla_\nu \sigma(GgG)^{\mu\nu}\del_\mu \phi_2
 + (GgG)^{\mu\nu}\nabla_\nu \del_\mu \phi_2\Big)
\,,
\eea
which implies \eqref{HXX-operator} since $\phi_1$ is arbitrary. 
Further simplification of this formula can be achieved in $2n=4$ dimensions, where 
\eqref{HXpXp-4d} follows directly from \eqref{HXpXp} using the 4D identity \eqref{4D-id-2}.
Hence in particular
\begin{align}
(2\pi)^2 \Tr \big(\Phi_2[X_a,(H^{ab} - \frac 12 H \eta^{ab}) [X_b,\Phi_1] ]\big)
&=- (2\pi)^2 \Tr \big((H^{ab} - \frac 12 H \eta^{ab}) [X_a,\Phi_2] [X_b,\Phi_1]\big)\nn\\
&\sim \intg e^\s g^{\mu\nu} \del_\mu \phi_2 \del_\nu \phi_1 \nn\\
&= -\intg  g^{\mu\nu} \phi_2 \nabg_\mu (e^\s \del_\nu \phi_1) 
\,, \label{TXXPP}
\end{align}
which for arbitrary $\phi_2$ implies \eqref{HXX-operator-4d}.
\end{proof}

Finally, we also note the following identity which will be useful below:
\begin{align}
H^{ab}[X_a,[X_b,\Phi]] &= [X_a,H^{ab}[X_b,\Phi]] - [X_a,H^{ab}][X_b,\Phi] \nn\\
&\sim  e^\s  \left(2\eta \Lap \phi - e^\s g^{\mu\nu}\nabG_\mu\del_\nu \phi
+2G^{\mu\nu} \del_\nu \eta  \del_\mu \phi  - \nabG_\nu (e^\s g^{\mu\nu}) \del_\mu \phi\right) \nn\\
&\quad - e^\s(e^\s\Lap x^c \del_\a x_c + \del_\a\eta) G^{\a\nu}\del_\nu\phi \nn\\
&=  e^{2\s}  \left(2e^{-\s}\eta \Lap \phi -g^{\mu\nu}\nabG_\nu\del_\mu\phi
+ (e^{-\s}G^{\mu\nu} \del_\nu \eta - g^{\mu\nu}\nabG_\nu \s)  \del_\mu \phi\right) 
\,. \label{XH-HXP}
\end{align}

\subsection{Potential terms \texorpdfstring{$k=0$}{k=0}}
\label{sec:potentials}

For $k=0$, consider first the following terms
\bea
\Tr \Big(\!-\inv{4}H\Big)^\ell \sim \intpiG e^{-\s}\eta^\ell \,, \qquad \textrm{for }\ell\in\N \,.
\eea
For $\ell=1$, we recover the basic Yang-Mills matrix model
\begin{align}
S_{YM} &= -\inv{4}\,\Tr H  \,\sim\, \intpiG\, e^{-\sigma}\eta  
\,.
\end{align}
Now recall that \eq{eta-sigma-alpha}
\be
e^{-\s}\eta= \inv2(\a^2 + \a^{-2}) \geq 1 \,,
\ee
which assumes its minimum $e^{-\sigma}\eta=1$ if and only if
$\a=\pm 1$, i.e. for $g_{\mu\nu}=G_{\mu\nu}$. 
This means that for fixed embedding, the minimum of the action $S_{YM}$ 
is achieved\footnote{This 
is certainly true in the Euclidean case, and in the Minkowski case provided 
we adopt complexified $\theta^{\mu\nu}$ as discussed in \secref{sec:4d-relations} 
and Ref.~\cite{Steinacker:2008ya}.} if $\a=\pm 1$, i.e.
if $\theta^{\mu\nu}$ is self-dual w.r.t. $g_{\mu\nu}$. 
Curvature terms as discussed below
may lead to small deviations from self-duality,
\be
G_{\m\n} = g_{\m\n} + h_{\m\n} \,,
\ee
however the potential is expected to dominate as long as the curvature is ``small''.
This is an important mechanism, which justifies to focus on 
geometries where $G_{\mu\nu} \approx g_{\mu\nu}$. The deviations from (anti-)self-duality
will be studied in more detail in \secref{sec:eom}; e.g. it will also be shown that $e^{-\s}\eta=1+\cO(h^2)$.

Thus assuming $G \approx g$, 
the above potential terms for $\ell>1$ amount to
\bea
\Tr \Big(\!-\inv{4}H\Big)^\ell \sim
 \intpiG e^{(\ell-1)\s}\, (e^{-\s}\eta)^\ell
\,\,\,\,\stackrel{g\approx G}{\approx}\,\,\,\, \intpiG e^{(\ell-1)\s}\, .
\eea
Then these terms essentially
determine a potential
\be
S_{\rm pot}   = \sum_\ell a_\ell \Tr H^\ell \,\,\,\stackrel{g\approx G}{\approx}\,\,\, \intpiG V(\s) \,,
\label{sigma-potential}
\ee
for $e^\s$.
This is very interesting: if $V(\sigma)$ has a non-trivial minimum, it will
dynamically determine the vacuum expectation value of $e^{\sigma}$ and 
hence the scale of non-commutativity.
Thus $e^\sigma$ will be essentially constant, simplifying considerably
some of the considerations below. 
This is also important in order to preserve the equivalence principle, at least approximately, 
because the effective metric for fermions and scalars a priori differ 
by a conformal factor $\sim e^{\sigma/3}$ \cite{Klammer:2008df,Klammer:2009dj}.
There are other terms with $k=0$ of  type $\Tr(H^{ab}H_{bc} H^{ca})$ etc.
For $g \approx G$, they essentially reduce to the same potential terms 
as above due to the projector
property $4 H^{ab}\eta_{bb'}H^{b'c} = H H^{ac}$ which holds for  $g_{\m\n}=G_{\m\n}$, assuming $2n=4$. 
However this type of terms 
also depends on the dimension of $\cM \subset \R^D$, and might help to single out 
4-dimensional branes. This should be investigated elsewhere. 
(In fact, $g_{\m\n}=G_{\m\n}$ is only possible for $2n=4$, which alone would single out 
4-dimensional branes.)

We can summarize these observations as follows: 
In the case of near-flat geometries the potential terms with $k=0$ are expected to dominate,
leading to $g_{\m\n} \approx G_{\m\n}$ and $e^\sigma \approx$ const.
Additional terms with $k>0$ involving more commutators typically correspond 
to curvature contributions as shown below, and
may lead to small deviations from $g=G$. In fact, 
it turns out that $\sigma = $ const. is incompatible with self-dual $\theta^{\mu\nu}$ 
resp. $g = G$ 
for general geometries\footnote{For example, such a self-dual $\theta^{\mu\nu}$
was determined for the Schwarzschild geometry in \cite{Blaschke:2010ye}, and it
turns out that $e^\sigma \neq {\rm const}$.}.
Nevertheless, the presence of a potential $V(\sigma)$ 
should ensure that $\sigma$ is constant to a very good
approximation, even in the presence of curvature. This is important 
because $e^\sigma$ determines e.g. the gauge coupling constant. It also suggests that 
the symplectic structure 
obtained in \cite{Blaschke:2010ye} based on self-duality will be modified
near the horizon, such that $e^\sigma \approx$ const. is preserved.
This should be studied in more detail elsewhere.

\subsection{\texorpdfstring{$\cO(X^{6})$}{O(X**6)} terms}

For the sake of systematics we start our discussion of  $k>0$ terms with the $\cO(X^{6})$, although the 
$\cO(X^{10})$ turn out to be much more appealing.
As shown in \cite{Blaschke:2010rg}, there are only two independent terms of order $X^6$, given by 
\begin{align}
S_6&= \Tr \left(\a \Box X^a \Box X_a + \frac{\b}{2} [X^c,[X^a,X^b]] [X_c,[X_a,X_b]]\right)
\,. \label{S6-geom}
\end{align}
In the general case $g_{\m\n} \neq G_{\m\n}$, it seems that the easiest way to evaluate them is in terms of 
$R[g]$ (also allowing us to compare with the one-loop results in \cite{Klammer:2009dj}).  
We start our derivation by considering
\begin{align}
\Box X^a&\sim -\th^{\m\n}\pa_\m x^b\nabg_\n\left(\th^{\a\b}\pa_\a x_b\pa_\b x^a\right)=-\left(e^\s G^{\n\b}\nabg_\n\pa_\b x^a+e^\s\th^{-1}_{\a\r}G^{\n\r}\nabg_\n\th^{\a\b}\pa_\b x^a\right) \nn\\
& = -e^\s \left(G^{\m\n}\nabg_\m\nabg_\n x^a-G^{\m\r}\nabg_\n x^a\th^{\s\n}\nabg_\m\th^{-1}_{\s\r}\right)\,,
\end{align}
since $\pa_\m x^b\nabg_\n\pa_\a x_b=0$. It follows that
\begin{align}
\Box X^a\Box X_a&\sim e^{2\s}\left(G^{\m\n}G^{\a\b}\nabg_\m\nabg_\n x^a\nabg_\a\nabg_\b x_a+e^\s G^{\m\r}G^{\a\t}G^{\s\e}\nabg_\m\th^{-1}_{\s\r}\nabg_\a\th^{-1}_{\e\t}\right)\nn\\
&=e^{\s}\left(e^\s G^{\m\n}G^{\a\b}\nabg_\m\nabg_\n x^a\nabg_\a\nabg_\b x_a+g_{\r\t}\nabg_\m G^{\m\r}\nabg_\a G^{\a\t}\right)
\,,
\end{align}
using \eqref{div-theta-g-id} for the second term (which is manifestly order $\cO(h^2)$). 
The second part of $S_6$ is derived in \appref{app:S6},
and using the 4D identity \eqref{4D-id-2} we find
\begin{align}
S_6&\sim (\a+\b)\intpig e^\s G^{\m\n}G^{\r\s}\nabg_\m\nabg_\n x^a\nabg_\r\nabg_\s x_a \nonumber\\
&\quad+\b\intpig e^\s\Bigg[2e^{-\s}\th^{\m\n}\hat\th^{\r\s}R[g]_{\m\r\n\s}- G^{\m\n}G^{\r\s}R[g]_{\m\r\n\s}-(GgG)^{\m\n}R[g]_{\m\n}\nn\\
&\quad\hspace*{3.5cm} +\left(\frac{3}{4}(Gg)G^{\m\n}-g^{\m\n}\right)\pa_\m\s\pa_\n\s\Bigg]\nonumber\\
&\quad+\a\intpig \nabg_\m G^{\m\r}\nabg_\a(Gg)^\a_\r \nn\\
&\quad+\b\intpig e^\s\Bigg[ 
\inv{2}\left(\frac{3}{2}G^{\m\n}\nabg_\m(Gg)-\nabg_\m G^{\m\n}(Gg)\right)\pa_\n\s 
+g^{\n\r}\nabg_\n\th^{-1}_{\t\r}\th^{\t\m}\pa_\m\s\nn\\* 
&\quad\hspace*{3.5cm} -\nabg_\m(Gg)^\n_\s G^{\r\s}\nabg_\n\th^{-1}_{\t\r}\th^{\t\m}
-\nabg_\m (Gg)^{\r}_\s \theta^{-1}_{\r\t}\nabg_\n G^{\t\m} \theta^{\n\s} \nn\\
&  \quad\hspace*{3.5cm} 
   + \frac 12 e^\s G^{\t\m}\nabg_\m G^{\r\r'}\theta^{-1}_{\n\r'} \th^{-1}_{\s\r} \nabg_\t G^{\n\s} \Bigg]
\label{eq:S6-semiclass}
\end{align}
where
\be
\hat \theta^{\mu\nu} := G^{\mu\mu'} g_{\mu'\eta} \theta^{\eta\nu}
\label{eq:def_theta-hat}
\ee 
is an anti-symmetric tensor. This is manifestly tensorial for $\a=-\b$.
Using \eqnref{eq:app-rel-lap-Lap-x} the first line of $S_6$ in the semi-classical limit \eqnref{eq:S6-semiclass} can also be written as
\begin{align}
\frac{\a+\b}{(2\pi)^2}\intg e^\s \mathcal{P}^{ab}_N\Lap x_a\Lap x_b\,.
\end{align}
The action \eqref{eq:S6-semiclass} simplifies considerably
in the self-dual case $g_{\m\n}=G_{\m\n}$, reducing to the one previously computed in Ref.~\cite{Blaschke:2010rg}. 
Furthermore, the terms surviving that limit are of the same type as those induced at one loop 
when coupling fermions to the matrix model, as was found in Ref.~\cite{Klammer:2009dj}. 
The leading order deviations from the self-dual case may be studied 
by expanding the above action around $G_{\m\n}=g_{\m\n}+h_{\m\n}$: 
To order $\cO(h)$ the action $S_6$ semi-classically reads
\begin{align}
S_6&\sim \frac{\a+\b}{(2\pi)^2}\intg e^\s \left(\lap x^a-2h^{\m\n}\nabg_\m\nabg_\n x^a\right)\lap x_a \nn\\
&\quad+\frac{\b}{(2\pi)^2}\intg e^\s\Bigg[2e^{-\s}\th^{\m\n}\left(\th^{\r\s}-h^{\r\a}g_{\a\b}\th^{\b\s}\right)R[g]_{\m\r\n\s} -2R[g]+4h^{\m\n}R[g]_{\m\n}\nn\\
&\quad\hspace*{3.5cm} +2\pa_\m\s\pa^\m\s -3h^{\m\n}\pa_\m\s\pa_\n\s
+2\nabg_\m h^{\m\n}\pa_\n\s-\frac{3}{4}\pa^\n(hg)\pa_\n\s\nn\\*
&\quad\hspace*{3.5cm} +g^{\n\r}\nabg_\n\th^{-1}_{\t\r}\th^{\t\m}\pa_\m\s
+\nabg_\m h^{\n\r}\nabg_\n\th^{-1}_{\t\r}\th^{\t\m}\Bigg]\nn\\*
&\quad +\cO(h^2)
\,. \label{eq:S6-order-h}
\end{align}
As explained in \secref{sec:eom}, $h_{\mu\nu}$ can be parametrized in terms of the 
deviation of the  symplectic structure around its self-dual version, 
i.e. $\th^{-1}_{\m\n}=\bar\th^{-1}_{\m\n}+F_{\m\n}$ where 
$\bar\th^{-1}_{\m\n}$ is self-dual with respect to $g_{\m\n}$. 
Then the above action can be simplified further by considering terms 
only up to order $\cO(F)$. This implies that $(h g) = \cO(F^2)$ can be dropped, and
$g^{\n\r}\nabg_\n\th^{-1}_{\t\r} = g^{\n\r}\nabg_\n F_{\t\r} = \cO(\del h)$.
The same type of matrix model terms have also been considered
on 2-dimensional branes in \cite{Arnlind:2010kw}, 
where $S_6$ for $\a+\b=0$ reduces essentially to an
integral over the Ricci scalar.

We also note that
\bea
R_{\mu\nu\rho\s}[g] \theta^{\mu\nu} \theta^{\rho\s} 
&=& W_{\mu\nu\rho\s}[g] \theta^{\mu\nu} \theta^{\rho\s} 
 - 2 e^{\s} R_{\mu\rho}[g] G^{\mu\rho} 
+\inv{3}e^\s R[g] (Gg) 
\eea
where~\cite{Wald:1984} 
\begin{align}
W_{\mu\nu\rho\s}&:=R_{\m\n\r\s}-\inv{2}(g_{\m\r}R_{\s\n}-g_{\m\s}R_{\r\n}-g_{\n\r}R_{\s\m}+g_{\n\s}R_{\r\m})+\inv{6}(Rg_{\m\r}g_{\s\n}-Rg_{\m\s}g_{\r\n})
\end{align}
is the Weyl tensor on the $4$ dimensional submanifold $\cM^4$. 
In the case of (anti-)self-dual $\theta$ we have $g=G$, and
\be
 R_{\mu\nu\rho\s}[g] \theta^{\mu\nu} \theta^{\rho\s} = W_{\mu\nu\rho\s}[g] \theta^{\mu\nu} \theta^{\rho\s} - \frac{2}{3} e^{\s} R[g] 
\,.
\label{Wthth}
\ee
This is interesting for the following reason: As discussed below, it may be appropriate to average over the
moduli space of Poisson structures $\theta^{\mu\nu}$, which essentially consists of (anti-) self-dual
2-forms with fixed determinant. This averaging over the asymptotic orientations leads to 
$W_{\mu\nu\rho\s}[g] \langle\theta^{\mu\nu} \theta^{\rho\s}\rangle = 0$
since $\langle\theta^{\mu\nu} \theta^{\rho\s}\rangle$ is Lorentz-invariant
for (A)SD $\theta$, so that the term 
$R_{\mu\nu\rho\s}[g] \theta^{\mu\nu} \theta^{\rho\s}$ essentially reduces to the Ricci scalar.

\subsection{\texorpdfstring{$\cO(X^{10})$}{O(X**10)} terms}

We now consider $\cO(X^{10})$ terms with $k>0$ (i.e. ignoring contributions to the 
 potential as discussed above in \secref{sec:potentials}). 
We are especially interested in a combination of terms which semi-classically more or less 
leads to the Einstein-Hilbert action. For $g_{\m\n}\neq G_{\m\n}$, the answer is not 
as simple as \eqnref{E-H-action-prev} derived in Ref.~\cite{Blaschke:2010rg}. 
As a starting point, we hence consider the term $H^{ab} \Box H_{ab}$ which previously 
has been shown to be the ``central piece'' leading to the Ricci-scalar in the semi-classical limit 
(i.e. the additional matrix terms were needed to make it intrinsic). 
The corresponding derivation is given in \appref{app:S10}. It reveals that the following 
combination of terms depend only on the intrinsic geometry through
$G_{\mu\nu},\,g_{\mu\nu}$ and $e^\s$, independent of the embedding $\cM^4 \subset \R^D$: 
\begin{align}
& H^{ab} \Box H_{ab} + 2\Box X^c H^{ab}[X_a, [X_b,X_c]]  \nn\\
&\sim - e^{3\sigma} 
\Big((GgGg) e^{-\s}\Box_G e^\sigma 
 + 8 e^{-2\s}(\del_\a \eta \del^\a \eta-\eta\del_\a \eta  \del^\a \s)
  - \frac 32\nabla_\nu g^{\mu\b}\nabla^\nu g_{\mu\b}
  + 2\nabla^{\mu}g_{\mu\b}\nabla_\a g^{\a\b} \nn\\
&\quad + (Gg)R_{\mu\eta}[G] (GgG)^{\mu\eta}  - 2R[G] 
 - 2\nabla_{\mu'}(G^{\mu\mu'}g_{\mu\b} \nabla_\a g^{\a\b})
- G_{\mu\b}\nabla_\a g^{\mu\rho}\nabla_\rho g^{\a\b}  \nn\\
& + 2\nabla_\b g^{\a\b}e^{-\s} \del_\a \eta 
- 4e^{-\s}\eta\nabla_{\mu}g^{\mu\a} \del_\a \s 
+ 2g^{\mu\nu}G_{\mu\a}\nabla_\b g^{\a\b}\del_\nu \s  \Big) 
\,. \label{HBoxH-bare}
\end{align}
The second term in the first line is needed in order to cancel extrinsic terms, 
and in the self-dual limit it semi-classically coincides with its counter part of 
Ref.~\cite{Blaschke:2010rg} (resp. the first term of \eqnref{E-H-action-prev}). 

In order to make the following results more transparent, we keep only terms of order 
$\cO(h)$ and drop higher-order terms in $h$. This is justified by the observation in 
\secref{sec:potentials} that the Yang-Mills action $S_{YM}$ is quadratic in $h$, 
and therefore suppresses the deviation 
from self-duality. Then the above result yields
\begin{align}
& \Tr(H^{ab} \Box H_{ab} + 2\Box X^c H^{ab} [X_a,[X_b,X_c]] ) \nn\\
& \quad\sim\,\, - \intpiG  e^{2\sigma} 
\Big( 4R_{\mu\eta}[G] (GgG)^{\mu\eta}  - 2R[G] 
+4 e^{-\s}\Box_G e^\sigma 
 + 4\nabla_\b g^{\a\b} \del_\a \s  \Big) +\cO(\del h^2) 
\,. \label{X10-curvature-action-1}
\end{align}
Using the intrinsic terms \eq{term-10a}, \eq{term-10c}, we also obtain the following 
forms
{\allowdisplaybreaks
\begin{subequations}\label{eq:S-grav}
\begin{align}
S_R&:=  \Tr(\inv2 T^{ab} \Box H_{ab} + \Box X^c H^{ab} [X_a,[X_b,X_c]] ) \nn\\*
&\sim - \intpiG  e^{2\sigma} 
\Big(3R[G] - 2R_{\mu\eta}[G] g^{\mu\eta}
 + 2\nabla_\b g^{\a\b} \del_\a \s  \Big) \quad + \cO(\del h^2)
\,, \label{S-grav-1}\\
\tilde S_R&:= \Tr(\inv2 H^{ab} (\Box H_{ab} - [X^a,[X^b,H]]) + \Box X^c H^{ab} [X_a,[X_b,X_c]]  ) \nn\\*
&\sim - \intpiG  e^{2\sigma} 
\Big(3R[G] - 2R_{\mu\eta}[G] g^{\mu\eta} \Big)  \quad + \cO(\del h^2)
\,, \label{S-grav-2}
\end{align}
\end{subequations}
noting that $\eta = e^\s + \cO(h^2)$ as well as
}
\bea
2R_{\mu\eta}[G] (GgG)^{\mu\eta}  - R[G] 
&=&  (4\e^{-\s}\eta -1)R[G] - 2R_{\mu\eta}[G] g^{\mu\eta} \nn\\
&=& 3R[G] - 2R_{\mu\eta}[G] g^{\mu\eta} \,\, + \cO(h^2 p^2) 
\,. \label{R-R-expand}
\eea
Here $p^2$ stands for the curvature scale of the gravitational field
 $R[G]$, and we will assume that $\cO(h^2 p^2) = \cO(\del h^2)$. 
For $G=g$, we recover the result obtained in \cite{Blaschke:2010rg}, 
and the ``local'' formula \eq{E-H-action-prev} follows from \eq{HBoxH-bare}.

\paragraph{Additional $\cO(X^{10})$ terms.}

Consider the following terms, whose
semi-classical limit is obtained easily from our previous results \eqref{eq:def-H}, \eqref{X-H-semiclass} and \eqref{dXboxX-4D}:
{\allowdisplaybreaks
\begin{subequations}
\bea
 [X^a,H_{ab}] [X^b,H] 
&\sim& 4 e^\s
(e^{\sigma}\nabla_\a g^{\b\a} \del_\b\eta  - G^{\a\b}\del_\a\eta  \del_\b\eta ) \nn\\*
&=& 4 e^{3\s}
(\nabla_\a h^{\b\a} \del_\b\s  - G^{\a\b}\del_\a\s  \del_\b\s )  + \cO(h^2)
\,, \label{term-10a}\\
\co{X^a}{H}\co{X_a}{H} &\sim& \, -16 e^\s \, G^{\mu\nu}\del_\mu \eta \del_\nu \eta  \nn\\*
&=& \, -16 e^{3\s} \, G^{\mu\nu}\del_\mu \s \del_\nu \s  + \cO(h^2) 
\,. \label{term-10c}
\eea
\end{subequations}
There are additional $\cO(X^{10})$ terms which are of order $\cO(h^2)$, 
which we will not discuss in this paper. These include
\begin{subequations}
\bea
H^{ab} \Box X^a \Box X^b 
&=& - e^{3\sigma} G_{\mu\nu}\nabla_\a g^{\a\mu}\nabla_\b g^{\b\nu}
\,\, = \,\, \cO(\del h^2)
\,, \label{term-10d}\\[1ex]
 \co{X_a}{H^{ab}} \co{X^c}{H_{bc}}
&\sim& - e^{2\s} 
\(\Box_G x^a \del_\a x_a + e^{-\s} \del_\a \eta\) \theta^{\a\b}\del_\b x^b
\(\Box_G x^c \del_\d x_c + e^{-\s}\del_\d \eta\) \theta^{\d\g}\del_\g x_b \nn\\*
&=& e^{3\s} 
\(-G_{\a\b}\nabla_\g g^{\a\g} \nabla_\k g^{\b\k}
 + 2  e^{-\s}\nabla_\a g^{\b\a}\del_\b \eta  
- e^{-2\s} G^{\a\b}\del_\a \eta \del_\b \eta \) \nn\\*
&=& \cO(\del h^2)  
\,. \label{term-10e}
\eea
\end{subequations}
The trace of the last term can in fact be written in a number of different ways,
}
\bea
\Tr \big( [X^a,H^{ab}] [X^c,H^{bc}] \big)
&=& - \Tr \big( H^{ab}[X^a,[X^c,H^{bc}]] \big) \nn\\
&=& \Tr \big( H^{ab}[X^c,[H^{bc},X^a]] \big) + \Tr \big( H^{ab}[H^{bc},[X^a,X^c]]] \big) \nn\\
&=& \Tr \big( [X^c,H^{ab}][X^a,H^{bc}] \big) + \Tr \big( [H^{ab},H^{bc}][X^a,X^c] \big) 
\,.
\eea

\paragraph{Extrinsic terms.}

The $\cO(X^{10})$ terms above have been tailored to be tensorial,
i.e. such that they only depend on the intrinsic geometry of $\cM$
in the semi-classical limit. 
There are of course also terms which depend on the ``extrinsic curvature''
i.e. on the embedding of $\cM \subset \R^D$. The prototype of such a term is
given by 
\bea
 \Box X^a \Box X_a
\,\,\sim\,\, e^{2\s} \Box_G x^a \Box_G x_a 
\,, \label{extrinsic-1}
\eea
cf. \eq{S6-geom},
or similarly
\bea
\Tr \Box X_a [X^b,(H_{bc} - \frac 12 \eta_{bc} H) [X^c,X_a]]
&\sim& \intpig e^\s \Box_G x^a (\Box_g x_a + g^{\mu\nu}\del_\mu\s \del_\nu x_a) \nn\\
&=& \intpig e^\s \big(\Lap x^a \lap x_a 
- g^{\m\n} G_{\n\a}  \nabG_\b g^{\a\b} \pa_\m\s \big)  \,, \nn\\
 H \Box X^a \Box X^b &\sim& -4 e^{2\s} \eta \Box_G x^a \Box_G x_a
\,. \label{ext-term-10}  
\eea
For $g_{\m\n} \sim G_{\m\n}$, these terms essentially coincide, and  
single out harmonic embeddings $\Box_G x^a = 0$ as vacuum geometries. 
In general, such terms should be expected to arise upon quantization, and 
their physical significance must be investigated. It seems plausible that they 
become important at cosmological scales where the intrinsic curvature is small, leading to 
long-distance modifications of gravity somewhat along the lines of 
the ``harmonic'' solutions given in \cite{Klammer:2009ku,Steinacker:2009mp}. 
Such long-distance modifications are 
very interesting in view of the major puzzles in cosmology, notably in the 
context of dark energy and dark matter.

On the other hand, the term $\Box_G x^a \Box_G x_a$ might also serve as a UV cutoff for
perturbation theory, since it behaves as $(p^2)^2$ on $\R^4_\theta$, 
where $p$ denotes the momentum scale.

\section{Gravitational action and degrees of freedom}
\label{sec:eom}

Now consider the matrix model action combining \eq{S-YM} with 
curvature terms such as \eq{eq:S-grav}, 
which in the semi-classical limit become
{\allowdisplaybreaks
\begin{subequations}\label{E-H-action-2}
\begin{align}
\tilde S_R 
&\sim   - \intpiG e^{2\sigma} 
(3R[G] - 2R_{\mu\eta}[G] g^{\mu\eta})  + \cO(\del h^2)  \nn\\*
&= -\intpig e^{2\sigma} 
(R[g] - 3 R^{\mu\nu}[g] h_{\mu\nu} + \nabg^\nu\nabg^\mu h_{\mu\nu})  + \cO(\del h^2)  
\,,  \\
S_{\textrm{simple}} &=
\tilde  S_R + \inv{2}\Tr[X^a,T_{ab}][X^b,H] \nn\\* 
&\sim  -\intpig e^{2\sigma} 
(R[g] - 3 R^{\mu\nu}[g] h_{\mu\nu} ) + \cO(\del h^2)  
\,, 
\end{align}
\end{subequations}
using \eq{RG-Rg-relation-h} where
}
\be
 G_{\mu\nu} = g_{\mu\nu} + h_{\mu\nu}
\ee 
and therefore $G^{\mu\nu} = g^{\mu\nu} - h^{\mu\nu} + \cO(h^2)$. 
The term $\nabg^\nu\nabg^\mu h_{\mu\nu}$ can be eliminated by subtracting 
suitable terms of type \eq{term-10a}, \eq{term-10c} from the action.  
We will therefore drop it and consider $S_{\textrm{simple}}$ 
in order to simplify the presentation. For the same reason the 
possible additional contributions from  $S_6$ \eq{eq:S6-order-h}
will also be omitted here.
We will furthermore  drop all terms of order $\cO(\del h^2)$, however 
we keep the $\cO(h^2)$ e.g. in the Yang-Mills terms and the potential terms, 
which are expected to be important for weak gravity. 
This will be justified below, and ensures a well-defined and compact moduli space
of vacuum solutions for $\theta^{\mu\nu}$.

Because these actions are tensorial 
(i.e. independent of the embedding $\cM^4 \subset \R^D$),
the semi-classical equations of motion are obtained simply 
by varying the independent geometrical degrees of freedom
encoded in $g_{\mu\nu}$ and $\theta^{\mu\nu}$.
To understand these degrees of freedom, 
note that in a given ``coordinate patch'', the  embedding metric 
$g_{\mu\nu} = \eta_{\mu\nu} + \del_\mu\phi^i\del_\nu\phi^j\eta_{ij}$ is determined by 
the scalar fields $\phi^i(x)$.
The Poisson tensor $\theta^{\mu\nu}$ can be parametrized as
\bea
\theta^{-1}_{\mu\nu} &=& \bar\theta^{-1}_{\mu\nu} + F_{\mu\nu}  
\eea
where $\bar\theta^{-1}_{\mu\nu}$ is self-dual\footnote{One could equally well consider
the case of small perturbations around anti-self-dual $\bar\theta^{-1}_{\mu\nu}$.} 
with respect to $g_{\mu\nu}$, and $F_{\mu\nu}= \del_\mu A_\nu - \del_\nu A_\mu$. 
Thus the independent degrees of freedom  are given by the embedding $\phi^i$ 
and $F_{\mu\nu}$ resp. $A_\mu$.

In principle, one could now 
derive the equations of motion resulting from \eq{E-H-action-2} as well as 
from the other possible terms such as $S_6$, \eqnref{eq:S6-order-h}. 
This is straightforward as long as only ``intrinsic'' terms 
are considered, which depend on $g_{\mu\nu}$ and $\theta^{\mu\nu}$.
The variation of the fundamental degrees of freedom can be separated into 
variations $\d_\phi$ of the embedding leading to
\be
\d_\phi g_{\mu\nu}=\d\phi^i \phi^j \eta_{ij} + \phi^i \d\phi^j \eta_{ij}
\,,
\ee
and the variation $\d_A$ of the Poisson tensor given by
\be
\d_A F_{\mu\nu} = \del_\mu \d A_\nu - \del_\nu \d A_\mu 
\,. 
\ee
We postpone this straightforward but tedious task to future work, 
and only draw some generic and qualitative conclusions below.
In the presence of terms which also depend on the embedding resp.
extrinsic curvature such as $\Lap x^a \Lap x^a$, the action would lead to 
higher-order equations of motion in the embedding $\phi^i$. In particular, this leads
to the ``harmonic branch'' as discussed in \cite{Steinacker:2009mp}, whose physical relevance requires 
further study. It may suffice here to say that such extrinsic terms may lead to 
very interesting cosmological solutions \cite{Klammer:2009ku}, while the viability for 
solar system gravity is not clear.

\paragraph{Yang-Mills action and vacuum configurations for $\theta^{\mu\nu}$.}

We can gain some important insights even without deriving equations of motion.
Let us expand the Yang-Mills term  to $\cO(F^2)$, but keep only 
$\cO(\del F)$ resp. $\cO(\del h)$ in the curvature terms
due to the explicit gravitational momentum scale. 
This gives
{\allowdisplaybreaks
\begin{subequations}\label{eq:vars-expand-in-F}
\begin{align}
\theta^{\mu\nu} &= ((1+\bar\theta F)^{-1} \bar\theta)^{\mu\nu} \nn\\*
&= (\bar\theta - \bar\theta F \bar\theta
 +  \bar\theta F \bar\theta F \bar\theta)^{\mu\nu} \,\, + \cO(F^3)\nn\\*
&= \bar\theta^{\mu\nu} + \bar\theta^{\mu\mu'}\bar\theta^{\nu\nu'} F_{\mu'\nu'} 
 +(\bar\theta F \bar\theta F \bar\theta)^{\mu\nu} \,\, + \cO(F^3) 
\,, \\
h_{\mu\nu} &= G_{\mu\nu} - g_{\mu\nu} 
= -e^{\bar\s} (\bar\theta^{-1} g F)_{\mu\nu}
 - e^{\bar\s} (F g \bar\theta^{-1})_{\mu\nu}
 - \frac 12 g_{\mu\nu}(\bar\theta F)  + \cO(F^2)
\,, \\
 g_{\mu\nu} h^{\mu\nu} &= 0 \, + \cO(F^2) \label{gh-contract} 
\,, \\
 h_{\mu\nu} h^{\mu\nu} 
&= 2 (\bar\theta F \bar\theta F) - 2 e^{\bar\s} (F g F g)
 - (\bar\theta F) (\bar\theta F) + \cO(F^3) 
\,, \\
e^{-\s} &= e^{-\bar\sigma}\det(1+\bar\theta F)^{1/2} \label{sigma-F} \nn \\*
&= e^{-\bar\sigma}(1+\frac 12 (\bar\theta F) + \frac 18 (\bar\theta F) (\bar\theta F)
 - \frac 14 (\bar\theta F\bar\theta F)  + \cO(F^3) ) 
\,, \\
(\theta g \theta g) 
&= -4 e^{\bar\s} + 2e^{\bar\s} (\bar\theta F)
 + e^{2\bar\s} (g F g F) - 2e^{\bar\s}(\bar\theta F \bar\theta F) + \cO(F^3)  
\,, \\
\frac 14 (Gg) &= - \frac 14 e^{-\s} (\theta g \theta g) 
= 1  +  \frac 14(\bar\theta F\bar\theta F)
 -\frac 18 (\bar\theta F) (\bar\theta F)- \frac 14 e^{\bar\s} (g F g F)+ \cO(F^3) \nn\\*
&= 1+\frac 18 h_{\mu\nu} h^{\mu\nu}+ \cO(F^3) 
\,. \label{theta-F-separation} 
\end{align}
\end{subequations}
Here we use a condensed notation where neighbouring indices are contracted and 
$()$ denotes a trace (e.g. $\bar\theta F \equiv \bar\theta^{\mu\nu} F_{\nu\eta}$ and 
$(\bar\theta F) \equiv \bar\theta^{\mu\nu} F_{\nu\mu}$), as well as
}
\begin{align}
e^{\bar\s}\d^\mu_\nu &= -  (\bar\theta g \bar\theta g)^\mu_\nu 
\,, &&
\nabg^\mu\bar\theta^{-1}_{\mu\nu} = 0 
\,. \label{div-theta-SD} 
\end{align}
The relation \eq{gh-contract} is in fact a consequence of $|G| = |g|$ in 4 dimensions, 
\eq{div-theta-SD} holds for any self-dual $\bar\theta^{-1}_{\mu\nu}$, and 
$\bar\s$ is defined through $\bar\theta^{\mu\nu}$, so that 
$\bar\cJ^\m_{\ \n}$ defines an almost-complex structure.
We will assume $h_{\m\n}$ to be small, and accordingly 
we will drop all terms of order $\cO(\del F^2)$.

The r.h.s. of \eq{theta-F-separation} 
acquires a geometric meaning due to the relation
\bea
\frac 18(F\bar\theta) (F\bar\theta) - \frac 14 (F\bar\theta F \bar\theta) 
&=& \textrm{Pfaff}(F_{\mu\nu}) \textrm{Pfaff}(\bar\theta^{\mu\nu})
\,, \label{FF-theta}
\eea
(cf. \cite{Steinacker:2007dq}) where
\bea
\textrm{Pfaff}(\bar\th^{\m\n}) &=& \inv8 \e_{\m\n\r\eta}\bar\th^{\m\n}\bar\th^{\r\eta} 
= \inv4 \inv{\sqrt{|g|}}\,\bar\th^{\m\n} (\star_g \bar\th)^{\m'\n'} g_{\m\m'}g_{\n\n'}
 = \pm \sqrt{|\bar\th^{\m\n}|}
\,. \label{pfaff-explicit}
\eea
Note that $\textrm{Pfaff}(\bar\th)$ is positive (negative) for (anti-) self-dual $\bar\th^{\m\n}$. 
Then the Yang-Mills matrix model action \eq{S-YM} in the semi-classical limit 
becomes\footnote{It is interesting
to compare this with the action for non-Abelian field strength \cite{Steinacker:2008ya}, which has a 
somewhat similar structure. The Abelian case has also been considered by
A. Schenkel (unpublished).}
\begin{align}
S_{YM} &\sim \intpig  e^{-\sigma}\eta 
= \intpig \big(1 +\inv8 h_{\mu\nu} h^{\mu\nu}  + \cO(F^3) \big)  \nn\\
&= \intpig \Big(1 + \inv4 e^{\bar\s} F_{\mu\nu} F_{\mu'\nu'} g^{\mu\mu'} g^{\nu\nu'} 
 - \textrm{Pfaff}(F_{\mu\nu}) \textrm{Pfaff}(\bar\theta^{\mu\nu}) + \cO(F^3) \Big) \nn\\
&= \intpig \Big(1 + \inv8 e^{\bar\s} (F\mp\star_g F)_{\m\n}(F\mp\star_g F)_{\m'\n'} g^{\m\m'} g^{\n\n'} \, + \cO(F^3) \Big)
\,, \label{S-YM-expand}
\end{align}
where $\star_g$ denotes the Hodge star w.r.t. $g_{\mu\nu}$, and $\mp$ is minus for self-dual $\bar\theta^{\mu\nu}$ 
and vice versa.

Recalling that any 2-form can be decomposed into self-dual (SD) and anti-selfdual (ASD) components, 
we arrive at an important result: 
ASD fluctuations $F_{\mu\nu}$ around a SD background $\bar\theta^{-1}_{\mu\nu}$ 
give a {\em positive} contribution to $S_{YM}$  and are hence suppressed, 
consistent with \eq{eta-sigma-alpha}. On the other hand, 
the SD part of $F_{\mu\nu}$ does {\em not} contribute to $S_{YM}$ but
determines the ``dilaton field'' $e^\s$.
Conversely, SD fluctuations around an ASD background are suppressed by $S_{YM}$, while  $e^\s$
encodes ASD fluctuations.
This justifies to focus on geometries with $G_{\mu\nu} \approx g_{\mu\nu}$, and makes clear that 
it is the embedding rather than the $\theta^{-1}_{\mu\nu}$ which plays the central role 
for the emergent gravity\footnote{It is nevertheless interesting to recall that 
this subject was sparked by the observation that the
$U(1)$ ``would-be'' gauge modes acquire a geometrical meaning through $G^{\mu\nu}$, leading to $h_{\mu\nu}$ 
which do give Ricci-flat fluctuations around flat backgrounds \cite{Rivelles:2002ez,Steinacker:2007dq}. 
This gauge sector is given a central role in~\cite{Yang:2010kj}. 
The ultimate physical relevance of these $U(1)$ modes is still to be understood.}.

In particular, it follows that the moduli space of 
vacuum configurations of $S_{YM}$ (for fixed embedding) 
consists of 2 disjoint components $\bar\Sigma = \bar\Sigma^+ \cup \bar\Sigma^-$
given by the space of (A)SD symplectic structures
$\bar\theta^{-1}_{\mu\nu}$ w.r.t. $g_{\mu\nu}$, and $S_{YM}$ provides a  positive definite action 
which suppresses fluctuations away from $\bar\Sigma$. 
These sectors $\bar\Sigma^\pm$ are disconnected, and characterized by the sign of $\textrm{Pfaff}(\bar\th^{\m\n})$.
Now observe that 
 $e^{\sigma}$ defines a scalar 
function on $\bar\Sigma$ \eq{sigma-F} which measures the ``strength'' of $\theta^{\mu\nu}$,
i.e. the non-commutativity scale. 
Hence a potential $V(e^\s)$ as in \eq{sigma-potential} 
with a non-trivial minimum, 
\begin{align}
V(e^\s) &= V_0 + \inv2 M^2(e^\s -x_0)^2  + \ldots \nn\\
  &= V_0 + \inv2 M^2 \big(e^{\bar\s}- x_0 -\inv2 e^{\bar\s}(\bar\th F) \big)^2  + \ldots 
\,, \label{V-explicit}
\end{align}
where $V_0$, $M$ and $x_0$ are constants, 
will set the NC scale resp. the vacuum scale $e^{\s} \approx$ const. 
Then $\Sigma$  becomes compact,
e.g. $\Sigma^{\pm} \cong S^2$ in the near-flat case.
On the other hand,  terms in the gravitational action such as $R_{\mu\nu} h^{\mu\nu}$
may lead to small deviations from (anti-)self-duality. Moreover, 
$e^\s =$ const. may not be compatible with 
(A)SD $\theta^{-1}_{\mu\nu}$ in the presence of curvature, cf. \cite{Blaschke:2010ye}.
Then \eq{V-explicit} suggests $(\bar\theta F) \approx 2 (1-e^{-\bar\s} x_0) \neq 0$ if $M$ is large, with
$F \to 0$ as $x \to \infty$. 
Therefore the physical moduli space $\Sigma = \Sigma^+ \cup \Sigma^-$ of vacua
will consist of symplectic forms 
$\theta^{-1}_{\mu\nu} = \bar \theta^{-1}_{\mu\nu} + F_{\mu\nu}$ 
which are small deformations of (A)SD fields,
characterized (in the asymptotically flat case) by the asymptotic orientation of
$\bar \theta^{-1}_{\mu\nu}$.

If the function $V(e^\s)$ has flat directions, 
then one can pick a vacuum with arbitrary scale $e^{\bar \s}$. The kinetic term $\del^\mu\s\del_\mu\s$
would still suppress variations of $\s$. 

We conclude that the above type of action represents a well-defined variational problem 
for the geometry, and leads to metrics with $g_{\mu\nu} \approx G_{\mu\nu}$ as well 
as $e^\s \approx$ const. Note that although we focused on the case of Euclidean signature, 
the steps go through in the Minkowski case provided one adopts 
complexified $\theta^{\mu\nu}$ as discussed above, which do admit (anti-)self-dual configurations 
$\star_g \theta = \pm i\theta$. 
This provides an important simplification and progress for the analysis of the emergent gravity theory.

\paragraph{Further perspectives and physical implications.}

One obvious class of vacuum solutions of  \eq{S-grav-2} and \eq{S-YM}  
is given by Ricci-flat spaces along with
an (A)SD $\theta^{\mu\nu}$ (hence $h_{\mu\nu} = 0$) such that $e^\s =$ const. 
The problem is that in general, Ricci-flat spaces may not admit such 
(A)SD $\theta^{\mu\nu}$ such that $e^\s =$ const. This is illustrated in 
\cite{Blaschke:2010ye} where a  self-dual $\bar\theta^{-1}_{\mu\nu}$
was found with $e^\s \neq$ const.

The above analysis suggests the following strategy to find  solutions
for the coupled system $(g_{\mu\nu},\theta^{\mu\nu})$:
for a given metric $g_{\mu\nu}$, compute first a self-dual symplectic form $\bar\theta^{-1}_{\mu\nu}$;
this will  lead to some $e^{\bar\s}$ which in general is not constant. 
Then $F_{\mu\nu}$ resp. $h_{\mu\nu}$ should be determined through the full 
equations of motion, which
will take the form of modified inhomogeneous Maxwell equations, schematically
\be
g^{\mu\mu'} \nabla^{\nu}(e^{\bar\s}F_{\mu'\nu})      
 =  J^\mu .
\ee
Here $J^\mu$ will depend on $\del_\nu \tilde V(e^{\bar\s})$ and $(\bar\theta F)$,
and may include matter contributions which turn out to act as dipole 
sources \cite{Steinacker:2010rh}.
In the presence of a suitable potential
$V(e^\s)$ and/or a kinetic term $\del^\mu\s\del_\mu\s$, this will lead to
$e^{\s} \approx$ const. 
Since the gauge coupling as well as the NC scale depends on $e^\sigma$, this
is probably essential to meet precision tests of general relativity 
and the time-independence of the fine structure constant.

The example of the Schwarzschild geometry \cite{Blaschke:2010ye} indicates a
certain tension between the requirements
$e^\s =$ const. and $g_{\mu\nu} = G_{\mu\nu}$, since $\theta^{\mu\nu}$ 
is determined by solving Maxwell-like equations with non-trivial boundary conditions.
This would presumably be acceptable if
$h_{\mu\nu} = \cO(R)$
for asymptotically flat 4-dimensional geometry, where $R$ denotes the scale of the 
gravitational curvature. 
In that case, the additional terms in 
the gravitational action such as $h^{\mu\nu} R_{\mu\nu} = \cO(R^2)$ are suppressed at least
in the case of weak gravity, leading to nearly-Ricci-flat spaces
$R_{\mu\nu} \approx 0$ as (vacuum) solutions
in agreement with  general relativity. However, this has not been shown at this point.

Even if the equations governing $\theta^{\mu\nu}$ are so rigid that
$h_{\mu\nu}$ cannot be neglected, one might still effectively recover an (almost)-constant
$e^\s$ along with (almost)-ASD $\theta^{-1}_{\mu\nu}$ e.g. by 
considering branes with compact extra dimensions, such as $\cM^4 \times S^2 \subset \R^{10}$. 
This is very natural also to obtain non-Abelian gauge groups as required
for particle physics (cf. \cite{Chatzistavrakidis:2010xi}), 
and will be studied elsewhere in more detail. 

There is another interesting point which should be kept in mind.
Once a solution for $\theta^{\mu\nu}$  is found, 
the quantization of the theory requires to 
integrate over the fluctuations in $F_{\mu\nu}$ (recall that this 
would-be $U(1)$ gauge field couples only to the gravitational sector).
However, there is in fact a moduli space $\Sigma$ of solutions $\theta^{\mu\nu}$, 
corresponding to different asymptotic orientations of $\theta^{\mu\nu}$ 
(this is obvious in the flat case). The question then arises whether 
one should also integrate over this moduli space\footnote{See also \cite{Doplicher:1994tu} 
for a related discussion in the context of {\nc} field theory.}.
In particular, this would amount to an integration over all 
configurations corresponding to different asymptotics of $\theta^{\mu\nu}$ related by 
Lorentz rotations. The Lorentz-violating term
$W\theta\theta$ \eq{Wthth} would then disappear from the action.
This issue boils down to the question whether or not there
really is a non-trivial VEV $\langle\theta^{\mu\nu}\rangle$,
spontaneously breaking  Lorentz invariance.
Note that this is not essential for the mechanism of gravity 
presented here, which works also (and in fact simplifies) under weaker assumptions
such as $\langle\theta^{\mu\nu}\rangle = 0$ but
$\langle\theta^{\mu\nu}\theta^{\mu'\nu'}\rangle \neq 0$.

Finally, we should perhaps comment on the cosmological constant
problem, which in the present setting amounts to 
explaining why $V' = 0$ implies $V \approx 0$,
i.e. that $V\approx0$ at its minimum (cp. \eqref{V-explicit}). 
At this stage (in the ``Einstein branch'' \cite{Steinacker:2010rh}) 
this problem may appear to be similar as in standard GR, but again there 
are additional ingredients such as extrinsic curvature, compact extra dimensions,
an additional (harmonic) branch of solutions, etc. 
which may shed new light on this problem.

\section{Concluding remarks}

The results of this paper represent a further step in the long-term project of 
studying the effective gravity theory emergent from matrix models of Yang-Mills type. 
One important new insight is that the ``bare'' Yang-Mills term defines a positive-definite action 
for $h_{\mu\nu} = G_{\mu\nu}-g_{\mu\nu}$, which implies that the effective metric
approximately coincides with the induced (embedding) metric. Furthermore, we studied the 
geometrical meaning of higher-order terms in the matrix model for general backgrounds, 
identifying in particular an action which is very similar to the Einstein-Hilbert action,
taking into account $G_{\mu\nu}\approx g_{\mu\nu}$ and $e^\s\approx$ const. 
Such terms are expected at the level of the quantum effective action, or alternatively
they can be added to the action by hand.
These results are very welcome in the quest for a realistic theory of (quantum) gravity.

We also identified some specific issues and potential problems in clarifying the 
physical viability and the relation with general relativity. 
One issue is a certain ``tension'' between self-dual $\theta^{\mu\nu}$ and $e^\s \approx$ const, 
which both seem natural and desirable in view of the above results. 
Once this is understood, one can
proceed to reliably analyze the equations for the embedding resp. for 
the effective metric, which 
then describes gravity and its deviation from GR.

The bottom line is that the model defines a highly non-trivial coupled system for 
the embedding  $g_{\mu\nu}$ and the Poisson structure $\theta^{\mu\nu}$, 
and contains some (quantum) theory of gravity. 
This complexity is of course essential for any serious candidate for a 
realistic theory, but makes the identification of the ``relevant'' configurations and
solutions  non-trivial. An additional complication is that quantum effects
must be taken into account, e.g. through higher-order terms as discussed here.
Furthermore, the case of compact extra dimensions and the implications of 
non-trivial extrinsic terms such as $\Lap x^a \Lap x^a$ must be 
studied systematically.
Clearly much more work is needed before the physical viability of 
these models can be reliably addressed.
 On the other hand, the models are sufficiently
clear-cut such that their physical content can finally be understood.

\subsection*{Acknowledgements}

This work was supported by the ``Fonds zur F\"orderung der Wissenschaftlichen Forschung'' (FWF) under contract P21610-N16.
H.S. is grateful for hospitality at the AEI Golm, and useful discussions with J. Arnlind, J. Hoppe, D. Oriti, 
M. Sivakumar, A. Schenkel, P. Schupp, R. Szabo and S. Theisen are acknowledged.

\startappendix

\Appendix{Derivation of \eq{CC-GGG} and \eq{nablaC}}
\label{app:CC-GGG}
Consider
{\allowdisplaybreaks
\bea
&& G^{\a\b} C_{\a;\sigma\nu} C_{\b;\mu\rho}G^{\rho\nu} G^{\sigma\mu} \nn\\*
&=& \frac 14 G^{\a\b} \(\nabla_\nu g_{\s\a} + \nabla_\s g_{\nu\a} -  \nabla_\a g_{\s\nu}\)
\(\nabla_\mu g_{\rho\b} + \nabla_\rho g_{\mu\b} -  \nabla_\b g_{\rho\mu}\)G^{\rho\nu} G^{\sigma\mu} \nn\\
&=& \frac 12 G^{\a\b} \(\nabla_\nu g_{\s\a}G^{\sigma\mu}  
- \frac 12  \nabla_\a g_{\s\nu}G^{\sigma\mu} \)
\(\nabla_\mu g_{\rho\b} + \nabla_\rho g_{\mu\b} -  \nabla_\b g_{\rho\mu}\)G^{\rho\nu} \nn\\
&=& \frac 12 \( \nabla_\nu (2e^{-\s}\eta G^{\b\mu} - g^{\b\mu})
- \frac 12G^{\a\b}\nabla_\a g_{\s\nu}G^{\sigma\mu} \)
G^{\rho\nu}\(\nabla_\mu g_{\rho\b} + \nabla_\rho g_{\mu\b} -  \nabla_\b g_{\rho\mu}\) \nn\\
&=&  \Big(G^{\b\mu}\del_\nu(e^{-\s}\eta)  - \frac 12\nabla_\nu g^{\b\mu}\Big)\nabla^\nu g_{\mu\b} 
 - G^{\a\b}\nabla_\a \Big(e^{-\s}\eta G^{\mu\rho} - \inv2 g^{\mu\rho}\Big)
\Big(\nabla_\rho g_{\mu\b} - \inv2 \nabla_\b g_{\rho\mu}\Big) \nn\\
&=& \frac 32  \del_\nu(e^{-\s}\eta)\del^\nu (Gg)   - \frac 34\nabla_\nu g^{\b\mu}\nabla^\nu g_{\mu\b} 
- \del_\a (e^{-\s}\eta) \nabla_{\mu} (2e^{-\s}\eta G^{\mu\a} - g^{\mu\a}) \nn\\*
&& +\frac 12  \nabla_\a g^{\mu\rho}\nabla_\rho (2e^{-\s}\eta \d_\mu^\a - G_{\mu\b}g^{\a\b})  \nn\\
&=& 4 \del_\nu(e^{-\s}\eta)\del^\nu (e^{-\s}\eta)  
 - \frac 34\nabla_\nu g^{\b\mu}\nabla^\nu g_{\mu\b} 
 + 2\del_\a (e^{-\s}\eta)  \nabla_{\mu}g^{\mu\a}
- \inv2 G_{\mu\b}\nabla_\a g^{\mu\rho}\nabla_\rho g^{\a\b} 
\label{CC-GGG-deriv}
\eea
assuming $2n=4$, where we have used \eqref{4D-id-2}. 

The relation \eq{nablaC} can be seen as follows:
\bea
&& g^{\s\mu}\nabla_\sigma C_{\mu;\rho\nu} - g^{\s\mu} \nabla_\rho C_{\mu;\s\nu} \nn\\*
&=& \frac 12 g^{\s\mu}\nabla_\sigma (\nabla_\rho g_{\mu\nu}+\nabla_\nu g_{\rho\mu}-\nabla_\mu g_{\rho\nu})
- \frac 12 g^{\s\mu} \nabla_\rho \nabla_\nu g_{\s\mu}\nn\\
&=& \frac 12 g^{\s\mu} \Big((\nabla_\rho\nabla_\sigma g_{\mu\nu}+\nabla_\nu \nabla_\sigma g_{\rho\mu}
 -\nabla_\sigma\nabla_\mu g_{\rho\nu})
-   \nabla_\rho \nabla_\nu g_{\s\mu}\nn\\*
&& +  ({R_{\s\rho\mu}}^\a g_{\a\nu} + {R_{\s\rho\nu}}^\a g_{\mu\a})
 + ( {R_{\s\nu\rho}}^\a g_{\a\mu} + {R_{\s\nu\mu}}^\a g_{\rho\a} )\Big)\nn\\
&=& \frac 12 g^{\s\mu} \Big(\nabla_\rho\nabla_\sigma g_{\mu\nu}+\nabla_\nu \nabla_\sigma g_{\rho\mu}
 -\nabla_\sigma\nabla_\mu g_{\rho\nu}
-   \nabla_\rho \nabla_\nu g_{\s\mu} \Big)\nn\\*
&& +  \frac 12 \Big(g^{\s\mu}{R_{\s\rho\mu\b}}[G] (Gg)^\b_{\nu} 
+ g^{\s\mu}{R_{\s\nu\mu\b}}[G] (Gg)^\b_{\rho}
- 2 R_{\a\rho\b\nu}[G] G^{\a\b}\Big)\nn\\
&=& \frac 12 \Big(-\nabla_\rho\nabla^\mu h_{\mu\nu}-\nabla_\nu \nabla^\mu h_{\rho\mu}
 +\Box_g h_{\rho\nu} 
 +  g^{\s\mu}  \nabla_\rho \nabla_\nu h_{\s\mu} 
\Big)\nn\\*
&& +  \frac 12 \Big(-R_{\rho\b}[g] h^{\b\a}g_{\a\nu} 
- R_{\nu\b}[g] h^{\b\a}g_{\a\rho}+ 2 R_{\a\rho\b\nu}[g] h^{\a\b}\Big) + \cO(h^2)
\,. \label{nablaC-general}
\eea
Now \eq{nablaC} follows noting that $g^{\rho\nu} \nabla\nabla h_{\rho\nu} = 0+\cO(h^2)$ 
due to \eq{eq:vars-expand-in-F}.
}

\Appendix{Semi-classical results for matrix model extensions}
\SubAppendix{Derivation of \eq{HXX-operator}}{app:HXX-op}
%
To see \eq{HXX-operator}, consider
\bea
H^{ab} [X_a,[X_b,\Phi]] 
&\sim& e^\sigma G^{\mu\nu}\del_\mu x^a \del_\nu x^b 
\theta^{\a\b}\del_\a x_a \del_\b(\theta^{\rho\eta}\del_\rho x_b\del_\eta\phi) \nn\\
&=& e^\sigma {(G g)^\nu}_\a \theta^{\a\b}
 \(\theta^{\rho\eta} g_{\nu\rho}\del_\b\del_\eta\phi
 +  \del_\b(\theta^{\rho\eta}g_{\r\n})\del_\eta\phi
 - \del_\b\del_\nu x^b \theta^{\rho\eta}\del_\rho x_b\del_\eta\phi\) \nn\\
&=& e^{2\sigma} (G g G)^{\b\eta}\del_\b\del_\eta\phi
 + e^\sigma \hat\theta^{\nu\b} \del_\b(e^{\sigma}G^{\eta\rho} \theta^{-1}_{\rho\nu}) \del_\eta\phi\nn\\
&=& e^{2\sigma} (G g G)^{\b\eta}\del_\b\del_\eta\phi
 + e^\sigma  \del_\b e^{\sigma} (G g G)^{\eta\b}  \del_\eta\phi\nn\\
&&+ e^{2\sigma}{(Gg)^{\b}}_\rho \del_\b G^{\eta\rho}  \del_\eta\phi
-\frac 12 e^{2\sigma} \hat\theta^{\nu\b} \del_\rho\theta^{-1}_{\nu\b} G^{\eta\rho}\del_\eta\phi 
\eea
using the fact that $\hat \theta^{\mu\nu}$ is anti-symmetric, and 
\bea
\hat\theta^{\nu\b} \del_\b\theta^{-1}_{\rho\nu} 
&=& -\hat\theta^{\nu\b} \del_\rho\theta^{-1}_{\nu\b} - \hat\theta^{\nu\b} \del_\nu\theta^{-1}_{\b\rho} 
\nn\\
2\hat\theta^{\nu\b} \del_\b\theta^{-1}_{\rho\nu} 
&=& -\hat\theta^{\nu\b} \del_\rho\theta^{-1}_{\nu\b}
\,.
\eea
On the other hand, consider 
\bea
(GgG)^{\mu\nu}\Gamma^\a_{\mu\nu}[G] 
&=& \frac 12 (GgG)^{\mu\nu}\(\del_\mu G_{\nu\b} + \del_\nu G_{\mu\b} - \del_\b G_{\mu\nu}\)G^{\a\b} \nn\\
&=&  -{(Gg)^{\mu}}_\b\del_\mu G^{\a\b}  - \frac 12(GgG)^{\mu\nu}\del_\b G_{\mu\nu} G^{\a\b}\nn\\
&=&  -{(Gg)^{\mu}}_\b\del_\mu G^{\a\b}  
+ \frac 12 (G^{\mu\nu}\del_\b g_{\mu\nu} 
+ 2 \hat\theta^{\a\b}\del_\mu \theta^{-1}_{\a\b} - (Gg)\del_\b \sigma)G^{\a\b}  \nn\\
&=&  -{(Gg)^{\mu}}_\b\del_\mu G^{\a\b}  
+ \frac 12 (\frac 12 \del_\b (Gg)
+ \hat\theta^{\a\b}\del_\mu \theta^{-1}_{\a\b} - \frac 12 (Gg) \del_\b \sigma)G^{\a\b}  \nn
\eea
using
\bea
\frac 12 \del_\mu (Gg) = G^{\a\b}\del_\mu g_{\a\b} 
+ \hat\theta^{\a\b}\del_\mu \theta^{-1}_{\a\b} 
- \frac 12 (Gg) \del_\mu \sigma
\,.
\eea
Therefore we get
\bea
H^{ab} [X_a,[X_b,\phi]] 
&\sim& e^{2\sigma} (G g G)^{\b\eta}\del_\b\del_\eta\phi
 + e^\sigma  \del_\b e^{\sigma} (G g G)^{\eta\b}  \del_\eta\phi\nn\\
&& - e^{2\sigma}((GgG)^{\mu\nu}\Gamma^\eta_{\mu\nu} 
- \frac 14 \del_\rho (Gg)G^{\eta\rho}
+\frac 14 (Gg) \del_\rho \sigma G^{\eta\rho} ) \del_\eta\phi \nn\\
&=&  e^{2\sigma} (G g G)^{\b\eta}\nabla_\b\del_\eta\phi
 + e^\sigma  \del_\b e^{\sigma} (G g G)^{\eta\b}  \del_\eta\phi
\nn\\ &&
 + \inv4 e^{2\s}  (\del_\rho (Gg) - (Gg) \del_\rho\sigma) G^{\eta\rho} \del_\eta\phi  \,,
\eea
which is indeed tensorial.

\SubAppendix{Derivation of \eq{eq:S6-semiclass}}{app:S6}
%
We use the (constant) background metric $\eta_{ab}$ to pull down Latin indices, i.e. $x_a\equiv x^b\eta_{ab}$, and consider first
\begin{align}
\inv{2}\co{X^c}{\co{X^a}{X^b}}\co{X_c}{\co{X_a}{X_b}}&\sim \inv{2}e^\s G^{\n\s}\nabg_\n\left(\th^{\a\b}\pa_\a x^a\pa_\b x^b\right)\nabg_\s\left(\th^{\t\e}\pa_\t x_a\pa_\e x_b\right)\nn\\
&=e^\s G^{\n\s}\left(e^\s G^{\a\t}\nabg_\n\nabg_\a x^a\nabg_\s\nabg_\t x_a+\inv{2}g_{\a\t}g_{\b\e}\nabg_\n\th^{\a\b}\nabg_\s\th^{\t\e}\right)
\,.\label{eq:S_6-zw1}
\end{align}
From the Jacobi identity
\begin{align}
\th^{\m\a}\nabg_\a\th^{\n\s}+\th^{\n\a}\nabg_\a\th^{\s\m}+\th^{\s\a}\nabg_\a\th^{\m\n}=0\,,
\end{align}
it follows that
\begin{align}
\nabg_\r\th^{\m\n}&=\left(\th^{\m\a}\th^{\n\s}-\th^{\m\s}\th^{\n\a}\right)\nabg_\a\th^{-1}_{\r\s}\,,
\end{align}
which enables us to simplify the second term of \eqref{eq:S_6-zw1} further:
\begin{align}
\inv{2}g_{\a\t}g_{\b\e}\nabg_\n\th^{\a\b}\nabg_\s\th^{\t\e}&=g_{\a\t}g_{\b\e}\nabg_\n\th^{\a\b}\th^{\t\m}\th^{\e\r}\nabg_\m\th^{-1}_{\s\r}\nn\\
&=\th^{\e\r}\nabg_\n\left(e^\s(Gg)^{\m}_{\e}\right)\nabg_\m\th^{-1}_{\s\r}+e^\s(Gg)^\r_\t\nabg_\n\th^{\t\m}\nabg_\m\th^{-1}_{\s\r}\nn\\
&=e^\s\left(\pa_\n\s\inv{2}\hat\th^{\m\r}\nabg_\s\th^{-1}_{\m\r}+\left(\nabg_\n(Gg)^{\m}_{\t}\th^{\t\r}+(Gg)^\r_\t\nabg_\n\th^{\t\m}\right)\nabg_\m\th^{-1}_{\s\r}\right)
\,,
\end{align}
where $\hat\th^{\m\n}:=(Gg)^{\m}_\e\th^{\e\n}$. Hence,
{\allowdisplaybreaks
\begin{align}
&\frac{(2\pi)^2}{2}\Tr\big(\co{X^c}{\co{X^a}{X^b}}\co{X_c}{\co{X_a}{X_b}}\big)\nn\\*
&\sim\intg e^\s G^{\n\s}\Big(G^{\a\t}\nabg_\n\nabg_\a x^a\nabg_\s\nabg_\t x_a+\pa_\n\s\inv{2}\hat\th^{\m\r}\nabg_\s\th^{-1}_{\m\r}\nn\\*
&\qquad\qquad+\left(\nabg_\n(Gg)^{\m}_{\t}\th^{\t\r}+(Gg)^\r_\t\nabg_\n\th^{\t\m}\right)\nabg_\m\th^{-1}_{\s\r}\Big)\nn\\
&=\intg e^\s \Big(G^{\n\s}G^{\a\t}\nabg_\n\nabg_\a x^a\nabg_\s\nabg_\t x_a+G^{\n\s}\pa_\n\s\inv{2}\hat\th^{\m\r}\nabg_\s\th^{-1}_{\m\r}\nn\\*
&\qquad\qquad+G^{\n\s}\nabg_\n(Gg)^{\m}_{\t}\th^{\t\r}\nabg_\m\th^{-1}_{\s\r}-\left(\pa_\m\s G^{\n\s}+\nabg_\m G^{\n\s}\right)(Gg)^\r_\t\th^{-1}_{\s\r}\nabg_\n\th^{\t\m}\nn\\*
&\qquad\qquad -G^{\n\s}\th^{-1}_{\s\r}\nabg_\m(Gg)^\r_\t\nabg_\n\th^{\t\m}-G^{\n\s}(Gg)^\r_\t\th^{-1}_{\s\r}\nabg_\m\nabg_\n\th^{\t\m}\Big)\nn\\
&=\intg e^\s \Big(G^{\n\s}G^{\a\t}\nabg_\n\nabg_\a x^a\nabg_\s\nabg_\t x_a
+G^{\n\s}\pa_\n\s e^{-\s}\nabg_\s\eta\nn\\*
&\qquad\qquad-G^{\n\s}\th^{-1}_{\s\r}\nabg_\n(Gg)^{\m}_{\t}\nabg_\m\th^{\t\r}+\left(\pa_\m\s G^{\n\s}+\nabg_\m G^{\n\s}\right)(Gg)^\r_\s\th^{-1}_{\t\r}\nabg_\n\th^{\t\m}\nn\\*
&\qquad\qquad -G^{\n\s}\th^{-1}_{\s\r}\nabg_\m(Gg)^\r_\t\nabg_\n\th^{\t\m}-(GgG)^{\n\r}\th^{-1}_{\r\t}\left(\co{\nabg_\m}{\nabg_\n}\th^{\t\m}+\nabg_\n(\pa_\m\s\th^{\t\m})\right)\Big)\nn\\
&=\intg e^\s \Big(G^{\n\s}G^{\a\t}\nabg_\n\nabg_\a x^a\nabg_\s\nabg_\t x_a
+\inv{4}G^{\n\m}\pa_\n\s \left(\nabg_\m(Gg)+(Gg)\pa_\m\s\right)\nn\\*
&\qquad\qquad-G^{\n\s}\th^{-1}_{\s\r}\left(\nabg_\n(Gg)^{\m}_{\t}\nabg_\m\th^{\t\r}+\nabg_\m(Gg)^\r_\t\nabg_\n\th^{\t\m}\right)\nn\\*
&\qquad\qquad+\pa_\m\s (GgG)^{\n\r}\th^{-1}_{\t\r}\nabg_\n\th^{\t\m}+\nabg_\m (Gg)^\n_\a G^{\r\a}\th^{-1}_{\t\r}\nabg_\n\th^{\t\m}\nn\\*
&\qquad\qquad -(GgG)^{\n\r}\th^{-1}_{\r\t}\left({R[g]_{\n\m\eta}}^\t\th^{\eta\m}+R[g]_{\n\eta}\th^{\t\eta}\right)\nn\\*
&\qquad\qquad -(GgG)^{\n\r}\left(\nabg_\n\pa_\r\s+\pa_\m\s\th^{-1}_{\r\t}\nabg_\n\th^{\t\m}\right)\Big)\nn\\
&=\intg e^\s \Big(G^{\n\s}G^{\a\t}\nabg_\n\nabg_\a x^a\nabg_\s\nabg_\t x_a
+\inv{4}G^{\n\m}\pa_\n\s \left(\nabg_\m(Gg)+(Gg)\pa_\m\s\right)\nn\\*
&\qquad\qquad -\nabg_\m(Gg)^\n_\s\nabg_\n\th^{-1}_{\t\r}\left(G^{\m\r}\th^{\s\t}+G^{\r\s}\th^{\t\m}\right)
-G^{\n\s}\th^{-1}_{\s\r}\nabg_\m(Gg)^\r_\t\nabg_\n\th^{\t\m}\nn\\*
&\qquad\qquad -2\pa_\m\s (GgG)^{\n\r}\nabg_\n\th^{-1}_{\t\r}\th^{\t\m}
 +\nabg_\n(GgG)^{\n\r}\pa_\r\s + (GgG)^{\n\r}\pa_\n\s\pa_\r\s\nn\\*
&\qquad\qquad +e^{-\s}\hat\th^{\n\b}R[g]_{\n\m\eta\b}\th^{\eta\m}-(GgG)^{\n\r}R[g]_{\n\r}\Big)\nn\\
&=\intg e^\s \Big(G^{\n\s}G^{\a\t}\nabg_\n\nabg_\a x^a\nabg_\s\nabg_\t x_a
+\inv{4}G^{\n\m}\pa_\n\s \left(\del_\m(Gg)+(Gg)\pa_\m\s\right)\nn\\*
&\qquad\qquad -\nabg_\m(Gg)^\n_\s G^{\r\s}\nabg_\n\th^{-1}_{\t\r}\th^{\t\m}
-\nabg_\m G^{\r\r'} \theta^{-1}_{\r'\t}\nabg_\n G^{\t\m} \theta^{\n\s}g_{\s\r} \nn\\*
&\qquad\qquad + \inv2 e^\s G^{\t\m}\nabg_\m G^{\r\r'}\theta^{-1}_{\n\r'} \th^{-1}_{\s\r} \nabg_\t G^{\n\s}
 -2 (GgG + \frac 12 g)^{\n\r}\nabg_\n\th^{-1}_{\t\r}\th^{\t\m}\pa_\m\s
 +\nabg_\n(GgG)^{\n\r}\pa_\r\s\nn\\*
&\qquad\qquad + (GgG)^{\n\r}\pa_\n\s\pa_\r\s +e^{-\s}\hat\th^{\n\b}R[g]_{\n\m\eta\b}\th^{\eta\m}-(GgG)^{\n\r}R[g]_{\n\r}\Big)
\end{align}
using \eqref{4D-id-2}, \eqref{nabla-theta-id}, \eqref{div-theta-g-id}, 
and the identities $\nabg_\s\eta=\inv2(g\th g)_{\m\n}\nabg_\s\th^{\m\n}$ and $(Gg)^\m_\a\th^{-1}_{\b\m}=-(Gg)^\m_\b\th^{-1}_{\a\m}$
as well as
\begin{align}
G^{\n\s}\th^{-1}_{\s\r}\nabg_\m(Gg)^\r_\t&\nabg_\n\th^{\t\m}
= -\nabg_\m G^{\r\r'} \nabg_\n(e^\s\theta^{-1}_{\r'\t}G^{\t\m}) G^{\n\s}\th^{-1}_{\s\r} \nn\\
=& e^\s\nabg_\m G^{\r\r'} \nabg_{\r'}\theta^{-1}_{\t\n}G^{\t\m} G^{\n\s}\th^{-1}_{\s\r} 
+ e^\s G^{\t\m}\nabg_\m G^{\r\r'} \nabg_\t\theta^{-1}_{\n\r'} G^{\n\s}\th^{-1}_{\s\r} \nn\\*
  &  -\nabg_\m G^{\r\r'} \theta^{-1}_{\r'\t}\nabg_\n (e^\s G^{\t\m}) G^{\n\s}\th^{-1}_{\s\r} \nn\\
 =& -\nabg_\m (G^{\r\r'}g_{\s\r}) \nabg_{\r'}\theta^{-1}_{\t\n}G^{\t\m} \theta^{\n\s}
 + e^\s\frac 12 G^{\t\m}\nabg_\m G^{\r\r'} \nabg_\t (\theta^{-1}_{\n\r'} G^{\n\s}\th^{-1}_{\s\r}) \nn\\*
&  -\nabg_\m G^{\r\r'} \theta^{-1}_{\r'\t}\nabg_\n (e^\s G^{\t\m}) G^{\n\s}\th^{-1}_{\s\r} 
 - e^\s\frac 12 G^{\t\m}\nabg_\m G^{\r\r'}\theta^{-1}_{\n\r'} \th^{-1}_{\s\r} \nabg_\t G^{\n\s} \nn\\
=& \nabg_\m (Gg)^{\n}_\s \nabg_{\n}\theta^{-1}_{\t\r}G^{\r\m} \theta^{\t\s} 
- \frac 12 G^{\t\m}\del_\m (Gg) \del_\t \s
 -e^\s \nabg_\m G^{\r\r'} G^{\t\m}\theta^{-1}_{\r'\t}\th^{-1}_{\s\r}  G^{\n\s}\del_\n \s \nn\\*
& -e^\s \nabg_\m G^{\r\r'} \theta^{-1}_{\r'\t}\nabg_\n G^{\t\m} G^{\n\s}\th^{-1}_{\s\r}
 - e^\s\frac 12 G^{\t\m}\nabg_\m G^{\r\r'}\theta^{-1}_{\n\r'} \th^{-1}_{\s\r} \nabg_\t G^{\n\s}\nn\\
=& \nabg_\m (Gg)^{\n}_\s \nabg_{\n}\theta^{-1}_{\t\r}G^{\r\m} \theta^{\t\s} 
- g^{\m\s'}\nabg_\m \theta^{-1}_{\s\s'} \theta^{\n\s} \del_\n \s \nn\\*
& -e^\s \nabg_\m G^{\r\r'} \theta^{-1}_{\r'\t}\nabg_\n G^{\t\m} G^{\n\s}\th^{-1}_{\s\r}
  - e^\s\frac 12 G^{\t\m}\nabg_\m G^{\r\r'}\theta^{-1}_{\n\r'} \th^{-1}_{\s\r} \nabg_\t G^{\n\s} 
\end{align}
where the last step follows from
\bea
e^\s \nabg_\m G^{\r\r'} G^{\t\m}\theta^{-1}_{\r'\t}\th^{-1}_{\s\r}  G^{\n\s}\del_\n \s 
&=& e^{-\s} \nabg_\m G^{\r\r'} (g_{\r'\t}\theta^{\t\m}) (g_{\s\r} \theta^{\n\s}) \del_\n \s \nn\\
&=& e^{-\s} \nabg_\m \Big(\inv2 (Gg) g_{\t\s} - G_{\t\s}\Big) \th^{\t\m}\th^{\n\s} \del_\n \s \nn\\
&=& -\frac 12  G^{\m\n} \del_\m (Gg)\del_\n \s + g^{\m\s'}\nabg_\m \theta^{-1}_{\s\s'} \theta^{\n\s} \del_\n \s \nn
\eea
using the 4D identity \eqref{4D-id-2}, since
\bea
e^{-\s} \nabg_\m  G_{\t\s}\theta^{\t\m}\theta^{\n\s} \del_\n \s
&=& e^{-\s} \nabg_\m (e^\s \theta^{-1}_{\t\t'} \theta^{-1}_{\s\s'} g^{\t'\s'})\theta^{\t\m}\theta^{\n\s} \del_\n \s \nn\\
&=& \nabg_\m \theta^{-1}_{\t\t'} g^{\t'\n}\theta^{\t\m} \del_\n \s 
- g^{\m\s'}\nabg_\m \theta^{-1}_{\s\s'} \theta^{\n\s} \del_\n \s 
- g^{\m\n} \del_\m \s \del_\n \s \nn\\
&=&  -\theta^{-1}_{\t\t'} g^{\t'\n}\nabg_\m\theta^{\t\m} \del_\n \s 
- g^{\m\s'}\nabg_\m \theta^{-1}_{\s\s'} \theta^{\n\s} \del_\n \s 
- g^{\m\n} \del_\m \s \del_\n \s \nn\\
&=& - g^{\m\s'}\nabg_\m \theta^{-1}_{\s\s'} \theta^{\n\s} \del_\n \s 
\eea
due to \eqref{nabla-theta-id}. 
Together with the definition of the curvature tensor with respect to the induced metric
\eqref{eq:def-Rg} we obtain \eq{eq:S6-semiclass}.
}

\SubAppendix{Derivation of $\cO(X^{10})$ terms}{app:S10}

Consider first
{\allowdisplaybreaks
\begin{align}
H^{ab} \Box H_{ab}
&\sim - e^{2\sigma} G^{\mu\nu}\del_\mu x^a \del_\nu x^b
\Box_G (e^\sigma G^{\mu'\nu'}\del_{\mu'} x_a \del_{\nu'} x_b) \nn\\*
&= - e^{2\sigma} G^{\mu\nu}\del_\mu x^a \del_\nu x^b
\Big(\Box_G e^\sigma G^{\mu'\nu'}\del_{\mu'} x_a \del_{\nu'} x_b
 + 2 e^\sigma G^{\mu'\nu'}\Box_G\del_{\mu'} x_a \del_{\nu'} x_b \nn\\*
 &\quad +4\del^\a e^\sigma G^{\mu'\nu'}\nabG_\a\del_{\mu'} x_a \del_{\nu'} x_b
+ 2e^\sigma G^{\mu'\nu'}\nabla_\a\del_{\mu'} x_a \nabla^\a\del_{\nu'} x_b
\Big)  \nn\\
&= - e^{3\sigma} 
\Big((GgGg) e^{-\s}\Box_G e^\sigma
 +2 G^{\mu\nu} G^{\mu'\nu'} G^{\a\b} C_{\mu;\a\mu'} C_{\nu;\b\nu'} \nn\\*
 &\quad + 2  (GgG)^{\m\m'}(\del_\mu x^a\nabla_{\mu'}\Box_G x_a 
+ R_{\mu\eta} (Gg)^{\eta}_{\mu'} + 2C_{\mu;\a\mu'}\del^\a \sigma )\Big)  \nn\\
&= - e^{3\sigma} 
\Big((GgGg) e^{-\s}\Lap e^\sigma  +2 G^{\mu\nu} G^{\mu'\nu'} G^{\a\b} C_{\mu;\a\mu'} C_{\nu;\b\nu'} \nn\\*
&\quad + 2  (GgG)^{\m\m'}(2C_{\mu;\a\mu'}\del^\a \s -G_{\mu\b} \nabG_{\mu'}\nabG_\a g^{\a\b}
-\nabla_{\mu'}\del_\mu x^a\Box_G x_a 
+ (Gg)^{\eta}_{\mu'} R_{\mu\eta}[G]) \Big) 
\end{align}
using \eq{dXboxX-4D}.
The second term is elaborated in \eq{CC-GGG}, and using the 4D identity \eq{4D-id-2},
 \eqref{cons-g-1} and \eq{g-det-nabla-id} we obtain
}
\begin{align}
& H^{ab} \Box H_{ab} 
\sim - e^{3\sigma} 
\Big((GgGg) e^{-\s}\Box_G e^\sigma  + \frac 12 \del_\nu(Gg)\del^\nu (Gg) + \del_\a (Gg)  \nabla_{\mu}g^{\mu\a} 
  - \frac 32\nabla_\nu g^{\b\mu}\nabla^\nu g_{\mu\b} \nn\\
&\quad + \big((Gg)G^{\mu\mu'} - 2g^{\mu\mu'}\big)
\big(2C_{\mu;\a\mu'}\del^\a \sigma 
-\nabla_{\mu'}\del_\mu x^a\Box_G x_a 
+ R_{\mu\eta} (Gg)^{\eta}_{\mu'} \big) \nn\\
 &\quad  - 2G^{\mu\mu'}g_{\mu\b}\nabla_{\mu'}\nabla_\a g^{\a\b}
 -G_{\mu\b}\nabla_\a g^{\mu\rho}\nabla_\rho g^{\a\b}\Big)  \nn\\
&= - e^{3\sigma} 
\Big((GgGg) e^{-\s}\Box_G e^\sigma 
 + \frac 12 \del_\nu(Gg)\del^\nu (Gg) + \del_\a (Gg)  \nabla_{\mu}g^{\mu\a} 
  - \frac 32\nabla_\nu g^{\b\mu}\nabla^\nu g_{\mu\b} \nn\\
&\quad - G_{\mu\b}\nabla_\a g^{\mu\rho}\nabla_\rho g^{\a\b} 
+ (Gg)\Big( \del_\a (Gg)\del^\a \sigma -\Box_G x^a\Box_G x_a 
+ R_{\mu\eta}[G] (GgG)^{\mu\eta} \Big) \nn\\
 &\quad - 2G^{\mu\mu'}g_{\mu\b} \nabla_{\mu'}\nabla_\a g^{\a\b}
+2g^{\mu\mu'}\nabla_{\mu'}\del_\mu x^a\Box_G x_a - 2R[G] \Big) 
\,.
\end{align}
Note that there are two terms 
$g^{\mu\mu'}\nabla_{\mu'}\del_\mu x^a\Box_G x_a$ and $\Box_G x^a\Box_G x_a$, which 
are not tensorial but depend on the embedding of $\cM^4 \subset \R^D$. 
They coincide in the self-dual case
where $g_{\m\n}=G_{\m\n}$, but in general they are independent. In order to 
obtain tensorial expressions, we must cancel these terms. 
This can be achieved using \eq{XH-HXP}:
\begin{align}
& H^{ab} \Box H_{ab} + 2\Box X^c H^{ab}[X_a, [X_b,X_c]]  \nn\\
\sim& - e^{3\sigma} 
\Big((GgGg) e^{-\s}\Box_G e^\sigma 
 + \frac 12 \del_\nu(Gg)\del^\nu (Gg) + \del_\a (Gg)  \nabla_{\mu}g^{\mu\a} 
  - \frac 32\nabla_\nu g^{\b\mu}\nabla^\nu g_{\mu\b} - 2R[G] \nn\\
&\quad - G_{\mu\b}\nabla_\a g^{\mu\rho}\nabla_\rho g^{\a\b} 
+ (Gg)\Big( \del_\a (Gg)\del^\a \sigma 
+ R_{\mu\eta}[G] (GgG)^{\mu\eta} \Big) 
 - 2G^{\mu\mu'}g_{\mu\b} \nabla_{\mu'}\nabla_\a g^{\a\b} \Big) \nn\\
& - 2 e^{3\s}\Box_G x^c\del_\mu x^c (e^{-\s}G^{\mu\nu} \del_\nu \eta - g^{\mu\nu}\nabla_\nu \s)   \nn\\
=& - e^{3\sigma} 
\Big((GgGg) e^{-\s}\Box_G e^\sigma 
 + \frac 12 \del_\nu(Gg)\del^\nu (Gg) 
+ \del_\a (Gg)  \nabla_{\mu}g^{\mu\a} 
  - \frac 32\nabla_\nu g^{\mu\b}\nabla^\nu g_{\mu\b} \nn\\
&\quad 
+ (Gg)R_{\mu\eta}[G] (GgG)^{\mu\eta}  - 2R[G] 
 - 2\nabla_{\mu'}(G^{\mu\mu'}g_{\mu\b} \nabla_\a g^{\a\b})
 + 2\nabla^{\mu}g_{\mu\b}\nabla_\a g^{\a\b} \nn\\
& - G_{\mu\b}\nabla_\a g^{\mu\rho}\nabla_\rho g^{\a\b} 
 +  (Gg)\del_\a (Gg)\del^\a \sigma 
+ 2\nabla_\b g^{\a\b}e^{-\s} \del_\a \eta 
- 2g^{\mu\nu}G_{\mu\a}\nabla_\b g^{\a\b}\del_\nu \s  \Big), 
\label{eq:appB-zw1}
\end{align}
where we also used \eq{dXboxX-4D}.
This is manifestly tensorial, and can be rewritten in various ways.
Under the integral, \eqref{eq:appB-zw1} can be simplified further using
\be
 \intG e^{2\s}\,
  \nabla_{\nu}(G^{\mu\nu}g_{\mu\b} \nabla_\a g^{\a\b}) 
= \intG e^{2\s}\,\(2 g^{\mu\nu}G_{\nu\eta}\nabG_\a g^{\a\eta} \del_\mu\sigma  
 - 4e^{-\s}\eta  \nabla_\a g^{\a\b} \del_{\nu}\s\) 
\,, \nn
\ee
so that
\begin{align}
&(2\pi)^2 \Tr (H^{ab} \Box H_{ab} + 2\Box X^c H^{ab}[X_a, [X_b,X_c]])  \nn\\
\sim& - \intG  e^{2\sigma} 
\Big((GgGg) e^{-\s}\Box_G e^\sigma 
 + \frac 12 \del_\nu(Gg)\del^\nu (Gg) 
+ \del_\a (Gg)  \nabla_{\mu}g^{\mu\a} 
  - \frac 32\nabla_\nu g^{\mu\b}\nabla^\nu g_{\mu\b} \nn\\
&\quad + (Gg)R_{\mu\eta}[G] (GgG)^{\mu\eta}  - 2R[G] 
 + 8 e^{-\s}\eta  \nabla_\a g^{\a\b} \del_{\nu}\s 
 + 2\nabla^{\mu}g_{\mu\b}\nabla_\a g^{\a\b} \nn\\*
&\quad - G_{\mu\b}\nabla_\a g^{\mu\rho}\nabla_\rho g^{\a\b} 
 +  (Gg)\del_\a (Gg)\del^\a \sigma 
+ 2\nabla_\b g^{\a\b}e^{-\s} \del_\a \eta 
- 6g^{\mu\nu}G_{\mu\a}\nabla_\b g^{\a\b}\del_\nu \s \Big) \nn\\
=&  - \intG  e^{2\sigma} 
\Big(4 e^{-\s}\Box_G e^\sigma 
 + 4 R_{\mu\eta}[G] (GgG)^{\mu\eta}  - 2R[G] 
 + 4 \nabla_\a g^{\a\b} \del_{\nu}\s 
\,\, + \cO(h^2) \Big) ,
\end{align}
noting that 
$(Gg) = 4 + \cO(h^2)$ due to \eq{gh-contract}, 
$(GgGg) = \frac 12 (Gg)(Gg) - 4$
and $\eta = e^\s + \cO(h^2)$.



\end{document}